%% file: main_form_cotraspli.tex
\documentclass[sn-mathphys-num]{sn-jnl}

\usepackage{graphicx}%
\usepackage{amsmath,amssymb,amsfonts}%
\usepackage{amsthm}%
\usepackage{mathrsfs}%
\usepackage[title]{appendix}%
\usepackage{xcolor}%
\usepackage{textcomp}%
\usepackage{manyfoot}%
\usepackage{booktabs}%
\usepackage{algorithm}%
\usepackage{algorithmicx}%
\usepackage{algpseudocode}%
\usepackage{listings}%
\usepackage{tikz}

\usepackage{comment}
\usepackage{makecell}
\usepackage{multirow}%
\usepackage{hhline}
\usepackage{todonotes}
\usepackage{fontawesome}

\def\N{\mathbb{N}}
\def\ta{\mathtt{a}}
\def\tb{\mathtt{b}}
\def\tc{\mathtt{c}}
\def\td{\mathtt{d}}

\def\tA{\mathtt{A}}
\def\tU{\mathtt{U}}
\def\tG{\mathtt{G}}
\def\tC{\mathtt{C}}

\DeclareMathOperator{\Fact}{Fact}

\DeclareMathOperator{\Pref}{Pref}

\DeclareMathOperator{\Suff}{Suff}

\DeclareMathOperator{\alphabet}{alph}
\DeclareMathOperator{\subseq}{Subseq}

\DeclareMathOperator{\enc}{enc}
\DeclareMathOperator{\binary}{bin}
\DeclareMathOperator{\decimal}{dec}

\newcounter{algno}
\newcommand{\alg}[1]{\refstepcounter{algno}\label{#1}}

\newcommand{\bigo}{\mathcal{O}}

\newcommand{\hpd}{%
\begin{tikzpicture}[scale=0.3]%
\draw (0,0) -- (1,0);%
\draw (0,0.2) -- (1,0.2);%
\draw (0.3,0) -- (0.3,0.2);%
\draw (0.6,0) -- (0.6,0.2);%
\draw (0.9,0) -- (0.9,0.2);%
\draw (1,0) arc (-150:150:0.2);%
\end{tikzpicture}
}

\newcommand{\hpdubi}{\overset{\infty*}{\hpd}}

\newcommand{\hpdb}{\hpd}
\newcommand{\hpdbe}{\mathrel{\hpd}}
\newcommand{\hpdbo}{\overset{1}{\hpd}}
\newcommand{\hpdboe}{\mathrel{\overset{1}{\hpd}}}
\newcommand{\hpdbp}{\overset{\mathtt{p}}{\hpd}}
\newcommand{\hpdbpe}{\mathrel{\overset{\mathtt{p}}{\hpd}}}

\newcommand{\hpdblin}{\overset{\mathtt{lin}}{\hpd}}

\newcommand{\hpdblino}{\overset{\mathtt{1'lin}}{\hpd}}
\newcommand{\hpdblinp}{\overset{\mathtt{p'lin}}{\hpd}}

\newcommand{\hpdbcon}{\overset{\mathtt{con}}{\hpd}}

\newcommand{\hpdbcono}{\overset{\mathtt{1'con}}{\hpd}}
\newcommand{\hpdbconp}{\overset{\mathtt{p'con}}{\hpd}}

\newcommand{\hpdblog}{\overset{\mathtt{log}}{\hpd}}
\newcommand{\hpdbloge}{\mathrel{\overset{\mathtt{log}}{\hpd}}}
\newcommand{\hpdblogo}{\overset{\mathtt{1'log}}{\hpd}}
\newcommand{\hpdblogp}{\overset{\mathtt{p'log}}{\hpd}}

\theoremstyle{plain}%
\newtheorem{theorem}{Theorem}
\newtheorem{proposition}[theorem]{Proposition}
\newtheorem{corollary}[theorem]{Corollary}
\newtheorem{lemma}[theorem]{Lemma}
\newtheorem{observation}[theorem]{Observation}

\theoremstyle{remark}%
\newtheorem{example}[theorem]{Example}%
\newtheorem{remark}[theorem]{Remark}%

\theoremstyle{definition}%
\newtheorem{question}{Question}
\newtheorem{problem}{Problem}

\begin{document}

\title[A Formalization of Co-Transcriptional Splicing]{A Formalization of Co-Transcriptional Splicing as an Operation on Formal Languages}

\author[1]{\fnm{Da-Jung} \sur{Cho}}\email{dajungcho@ajou.ac.kr}

\author[2]{\fnm{Szil\'ard Zsolt} \sur{Fazekas}}\email{szilard.fazekas@ie.akita-u.ac.jp}

\author[3]{\fnm{Shinnosuke} \sur{Seki}}\email{s.seki@uec.ac.jp}

\author[4]{\fnm{Max} \sur{Wiedenh\"oft}}\email{maw@informatik.uni-kiel.de}

\affil[1]{\orgdiv{Department of Software and Computer Engineering}, \orgname{Ajou University}, \orgaddress{\street{206 World cup-ro}, \city{Suwon-si}, \state{Gyeonggi-do}, \postcode{16499}, \country{Republic of Korea}}}

\affil[2]{\orgdiv{Graduate School of Engineering Science}, \orgname{Akita University}, \orgaddress{\street{1-1 Tegatagakuen-machi}, \city{Akita City}, \postcode{010-0852}, \state{Akita}, \country{Japan}}}

\affil[3]{\orgname{University of Electro-Communications}, \orgaddress{\street{1-5-1 Chofugaoka}, \city{Chofu}, \postcode{1828585}, \state{Tokyo}, \country{Japan}}}

\affil[4]{\orgdiv{Department of Computer Science}, \orgname{Kiel University}, \orgaddress{\street{Christian-Albrechts-Platz 4}, \city{Kiel}, \postcode{24118}, \state{Schleswig-Holstein}, \country{Germany}}}

\abstract{ 
RNA co-transcriptionality is the process where RNA sequences are spliced while being transcribed from DNA templates. This process holds potential as a key tool for molecular programming. Co-transcriptional folding has been shown to be programmable for assembling nano-scale RNA structures, and recent advances have proven its Turing universality. While post-transcriptional splicing has been extensively studied, co-transcriptional splicing is gaining attention for its potential to save resources and space in molecular systems. However, its unpredictability has limited its practical applications.
In this paper, we focus on engineering co-transcriptional splicing, moving beyond natural occurrences to program RNA sequences that produce specific target sequences through DNA templates. 
We introduce contextual lariat deletion operations under three energy models—linear loop penalty, logarithmic loop penalty, and constantly bounded loop length—as well as bracketed contextual deletion, where deletion occurs solely based on context matching, without any structural constraints from hairpin loops.
We examine the complexity of the template constructability problem associated with these operations and study the closure properties of the languages they generate, providing insights for RNA template design in molecular programming systems.

}

\keywords{
Transcription,
RNA co-transcriptionality, Co-transcriptional splicing, 
Hairpin deletion operations, 
Formal Model,
DNA template design
}

\maketitle
 
    \section{Introduction}
    \label{sect:intro}
\input{sec-1-introduction}

    \section{Preliminaries}
    \label{sect:preliminaries}
\input{sec-2-prelims}

    \section{Template Construction Problem}
    \label{sect:temp-cons-hpd}
\input{sec-3-new-temp-cons}

    \section{Computational power of co-transcriptional splicing}
    \label{sect:comp-pow-hpd}
\input{sec-4-new-comp-prop}

    \section{Conclusion}
    \label{sect:conclusion}
\input{sec-5-conclusion}

\bibliography{main_form_cotraspli}

\end{document}

%% file: sec-1-introduction.tex
RNA co-transcriptionality refers to any process that an RNA sequence undergoes while it is being synthesized from its DNA template (\textit{transcribed}). This phenomenon plays a fundamental role in molecular computing in nature and is expected to become a key primitive in molecular programming.
Indeed, Geary, Rothemund, and Andersen have proved that one of such processes called co-transcriptional folding is programmable for assembling nano-scale single-stranded RNA structures \textit{in vitro}~\cite{GearyRA14} and then Seki, one of the authors, proved in collaboration with Geary and others that it is Turing universal by using an \textit{oritatami} model~\cite{GearyMSS19}. 

\begin{figure}[b]
    \centering
    \includegraphics[scale=0.5]{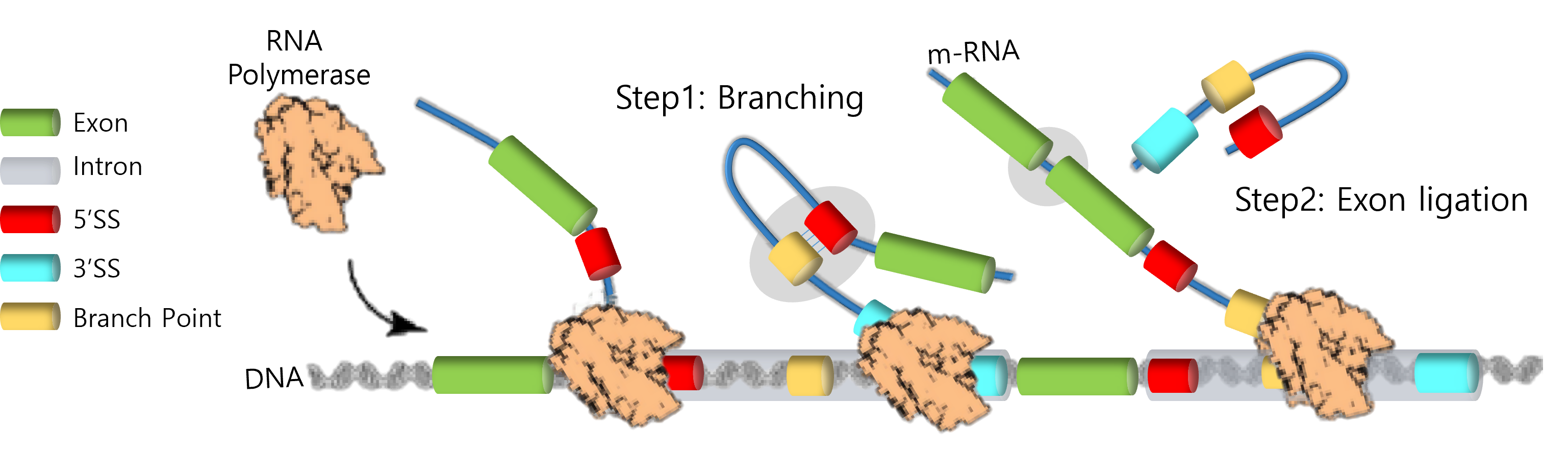}
    \caption{An illustration of co-transcriptional splicing: Introns are excised from primary transcripts by cleavage at conserved sequences known as splice sites~(SS), located at the 5’ and 3’ ends of introns. The splicing process initiates with the transcript folding into a hairpin shape, binding to the 5’SS along with the downstream branch point. The hairpin is then removed along with the remaining transcript, referred to as a lariat, at the 3’SS. Consequently, the exons are linked covalently, and the lariat containing the intron is released.}
    \label{fig:splicing}
\end{figure}

In addition to folding, RNA sequences can also undergo co-transcriptional splicing~\cite{MerkhoferHJ14}, where a subsequence is excised during transcription if it is flanked by specific motifs known as 5'- and 3'-splice sites~(SS). 
Figure~\ref{fig:splicing} illustrates co-transcriptional splicing.
As soon as transcribed, the 5'-SS is recognized by the polymerase-spliceosome complex and kept nearby, awaiting hybridization with its complementary sequence upon subsequent transcription.
If the 3'-SS is transcribed immediately after, the factor is excised.
It is much more extensively studied how RNA sequences are edited after being fully transcribed. 
This type of post-transcriptional splicing is non-deterministic because the final RNA product can vary depending on how the spliceosome interacts with different splice sites, generating distinct mRNA variants and enabling multiple RNA sequences to be encoded from a single template.

Co-transcriptional splicing has significant potential for economizing both resources and space in programming molecular systems. 
However, co-transcriptionality, with its inherent drive for efficiency and predictability, may be too greedy to accommodate the variability introduced by non-determinism. RNA sequences are widely utilized as key materials in molecular systems.
Due to their higher cost of commercial synthesis, systems are provided not with RNA sequences \textit{ab initio} but with their DNA templates instead along with polymerases for run-time synthesis of the RNA sequences. 
In this setting, co-transcriptional splicing offers an appealing mechanism for generating multiple RNA products from a single DNA template in a compact and predictable manner. This motivates the central question of this work: \emph{Can we encode a (finite) set of RNA sequences of a given system onto a single DNA template in such a way that all the sequences can be obtained by having the template be spliced co-transcriptionally, possibly along with those which do no harm to the system?} 
In this paper, we study this template construction problem in a formal model of co-transcriptional splicing. 
Our work aligns with the tradition of formal language theory that investigates computational properties of languages under biologically inspired operations~\cite{DomaratzkiO04, KariT96, KariT95}.

Co-transcriptional splicing remains far from being fully understood~\cite{HornGLEL23,SanchezGR22}. 
A factor forms a hairpin when a specific region of RNA folds back on itself, creating a loop of unbound bases and a stem of stacked base pairs. 
In nature, it remains unclear whether this hairpin formation is essential for co-transcriptional splicing or merely a by-product of splice site recognition. We propose its necessity, as the 5' splice site (5'-SS) alone is not sufficiently long to maintain the required proximity between the upstream and downstream regions of the RNA for ligation without the stabilization provided by the hairpin structure.
Co-transcriptionality is primarily governed by kinetic factors and operates in a greedy manner, acting on splice sites as soon as they are transcribed. 
As splicing decisions depend more on local interactions than on achieving the most stable configuration, the model may not need to take full account of thermodynamics. 
A hairpin results from having complementary prefix and suffix, say $x$ and $\theta(x)$, of a sequence of the form $x\ell \theta(x)$ hybridize with each other and the infix $\ell$ serve as a loop. 
Each interface between adjacent base pairs in the stem contributes a constant energy value, determined by the surrounding four bases, which reduces the overall free energy of the hairpin structure.
The stability of such structures is influenced by both components: a longer stem generally increases stability due to base pairing and stacking interactions, while a longer loop tends to destabilize the structure. Based on these characteristics, we consider three models for evaluating hairpin stability: \emph{linear loop penalty} where loop destabilization grows linearly with its length, \emph{logarithmic loop penalty} where it grows logarithmically, and 
\emph{constantly-bounded loop-length} where the loop length is restricted by a fixed constant. These models reflect different assumptions about how local RNA folding affects splicing decisions and provide a basis for our formal definition of co-transcriptional hairpin deletion operations.
Taking into account the hairpin models, we introduce operations of \emph{contextual lariat deletions} as a computation of co-transcriptional splicing.
We also define \emph{bracketed contextual deletion}, where the deletion occurs based solely on splicing site recognition, without considering the influence of hairpin structures.
Finally, we introduce the notions of ``parallel”, and ``iterated" operations for each bracketed contextual deletion and contextual lariat deletions.

Returning to the problem of designing DNA templates to produce a given set of RNA sequences, we investigate the complexity of the decision problems related to bracketed contextual deletion and contextual lariat deletion operations and provide insights into this process, which could lead to approximations for template design in specific practical problems. Table~\ref{tab:resultsofproblems} summarizes the computational complexity of the template constructibility problems for iterated bracketed contextual deletion~$\to_{[C]}^*$, 1-step of contextual lariat deletion~$\hpdbo_{\lbrack C \rbrack}$, and parallel lariat deletion~${\hpdbp_{\lbrack C \rbrack}}$.
\begin{table}[h!]
\centering

\begin{tabular}{ | c | c | c | c | } 
 \hline
 \mbox{} & $\to_{[C]}^*$ & $\hpdbo_{[C]}$ & $\hpdbp_{[C]}$ \\
 \hline
 Problem~\ref{problem:construct-exact-match}    
    & \makecell{P \\(Prop.~\ref{prop:problem-1-P})} 
    & {\color{gray}Not-Considered}  
    & {\color{gray}Not-Considered} \\ \hline

 Problem~\ref{problem:main-construction}  
    & \makecell{NP-Complete ($|\Sigma|\geq 4$) \\ (Theorem~\ref{thm:main-prob-decision-hard})}   
    & ??            
    & \makecell{NP-Complete ($|\Sigma|\geq 4$) \\ (Prop.~\ref{prop:parallelconstructability})} \\ \hline

 Problem~\ref{problem:verification}         
    & \makecell{P \\(Prop.~\ref{prop:problem-1-P})}                                  
    & \makecell{P \\ (Prop.~\ref{prop:hpa-verification-in-p})}                 
    & \makecell{P \\ (Prop.~\ref{prop:hpa-verification-in-p})} \\ 
 \hline
\end{tabular}
\vspace{2mm}
\caption{Summary of results for the template construction problems that ask whether a single DNA template can be engineered to generate a prescribed family of RNA words through co-transcriptional splicing: exact constructability~(Problem~\ref{problem:construct-exact-match}), more practical length-bounded constructability~(Problem~\ref{problem:main-construction}), and verification~(Problem~\ref{problem:verification}). Each result for the 1-step contextual lariat deletion and parallel lariat deletion holds across all three considered models: the linear loop penalty model, the constantly bounded loop-length model, and the logarithmic loop penalty model.}
\label{tab:resultsofproblems}
\end{table}

This motivates a study of the languages generated by these operations, with particular focus on their closure properties.
Table~\ref{tab:resultsofclosure} provides a summary of the computational power of bracketed contextual deletion~$\to_{[C]}$, and contextual lariat deletion~$\hpdbe_{\lbrack C \rbrack}$, where the latter is defined under three distinct energy models: the linear loop penalty model~$\hpdblin_{\lbrack C \rbrack, (d_1, d_2, d)}$, the constantly bounded loop-length model~$\hpdbcon_{\lbrack C \rbrack, (d)}$, and the logarithmic loop penalty model~$\hpdblog_{\lbrack C \rbrack, (d_1, d_2, d)}$.

\begin{table}[h]
\centering

\begin{tabular}{ | c | c | c | c | c | } 
 \hline
  & $\to_{[C]}^*$ & $\hpdblin_{[C],(d_1,d_2,d)}$ & $\hpdbcon_{[C],(c)}$ & $\hpdblog_{[C],(d_1,d_2,d)}$ \\
 \hline
 Regular        & \makecell{REG \\(Prop.~\ref{prop:regclosure})}       
                & \makecell{REG\\(Lemma~\ref{lemma:linear-hairpin-deletion-regular-languages-closed})}       
                & \makecell{REG\\(Corollary~\ref{corollary:regularclosed})}     
                & \makecell{REG\\(Corollary~\ref{corollary:reg}) } \\ \hline 

 1-Counter      & \makecell{Undecidable \\ (Prop.~\ref{prop:onecounternotclosed})}        
                & ??      
                & ??      
                & ?? \\ \hline

 Linear         & \multirow{2}{*}{%
                   \parbox[c][3.8em][c]{2.9cm}{\centering Undecidable\\(Prop.~\ref{prop:linclosure})}}  
                & \makecell{Not closed\\(Prop.~\ref{prop:linear-haipin-deletion-liner-languages-not-closed})}        
                & ??      
                & \makecell{Not closed\\(Corollary~\ref{corollary:notclosedcfl})} \\
 \hhline{|-|~|---|} 
 Context-Free   && \makecell{Not closed\\(Lemma~\ref{lemma:H-lin-not-context-free})}        
                & \makecell{Undecidable\\(Theorem~\ref{theorem:constant-bound-cf-to-undec})}    
                & \makecell{Not closed\\(Corollary~\ref{corollary:notclosedcfl})} \\
 \hline
\end{tabular}
\caption{Summary of the closure properties of regular, 1-counter, linear, and context-free languages under bracketed contextual deletion and contextual lariat deletion operations.}
\label{tab:resultsofclosure}
\end{table}

\newpage

%% file: sec-2-prelims.tex
Let $\N$ denote the set of positive integers and let $\N_0 = \N\cup\{0\}$. 
Let $\mathbb{Z}$ denote the set of integers. For some $m\in\N$, denote by $[m]$ the set $\{1,2,...,m\}$ and let $[m]_0$ = $[m]\cup\{0\}$. Let $\Sigma$ be a finite set of letters, called an \emph{alphabet}.
A \emph{word} over $\Sigma$ is a finite sequence of letters from $\Sigma$. 
With $\varepsilon$, we denote the \emph{empty word}. 
The set of all words over $\Sigma$ is denoted by $\Sigma^*$. Additionally, set $\Sigma^+ = \Sigma^*\setminus\{\varepsilon\}$. 
The \emph{length} of a word $w\in\Sigma^*$, i.e., the number letters in $w$, is denoted by $|w|$; hence, $|\varepsilon| = 0$. 
For some $w\in\Sigma^*$, if we have $w = xyz$ for some $x,y,z\in\Sigma^*$, then we say that $x$ is a \emph{prefix}, $y$ is a \emph{factor}, and $z$ is a \emph{suffix} from $w$. 
Being different from $w$ (strictly shorter), they are said to be \textit{proper}. 
We denote the set of all factors, prefixes, and suffixes of $w$ by $\Fact(w)$, $\Pref(w)$, and $\Suff(w)$, respectively. A word $u\in\Sigma^*$ is a subsequence of $w$ if $u$ can be obtained from $w$ by deleting arbitrary letters. The set of all subsequences of $w$ is denoted by $\subseq(w)$.
By $w[i]$ we mean the $i$-th letter of $w$ and $w[i..j]=w[i]\cdots w[j]$. 
For each regular expression $r$, writing it down, immediately refers to its language, e.g., writing $\ta(\ta|\tb)^*\tb$ refers to the set $\{\ta w\tb \mid w\in \{\ta,\tb\}^*\}$. 
We refer to languages obtained by pushdown automata with unary stacks (i.e., stacks with an alphabet of size 1) as 1-counter languages.

RNA is a chemically-oriented chain of nucleotides (letters) of four types: {\tt A}, {\tt C}, {\tt G}, and {\tt U}. 
Two RNA sequences of the same length, $x = a_1a_2 \cdots a_n$ and $y = b_n \cdots b_2 b_1$ can \textit{hybridize in the antiparallel manner} with each other if $a_i$ can bind to $b_i$ for all $1 \le i \le n$ (i.e., the first letter of $x$ with the last one of $y$, the second letter of $x$ with the second last one of $y$, and so on). 
Hybridization can be modeled using an \textit{antimorphic involution}, which is a function $\theta: \Sigma^* \to \Sigma^*$ that satisfies $\theta(\theta(a)) = a$ for all $a \in \Sigma$ and $\theta(wb) = \theta(b)\theta(w)$ for all $b \in \Sigma$ and $w \in \Sigma^*$. 
For instance, if $\theta$ maps {\tt A} to {\tt U} and vice versa, and {\tt C} to {\tt G} and vice versa (so-called Watson-Crick complementarity), then $\theta({\tt ACC}) = \theta({\tt C})\theta({\tt AC}) = {\tt G}\theta({\tt C})\theta({\tt A}) = {\tt GGU}$, and these two RNA sequences hybridize with each other as $\begin{array}{ccccc}
     {-} & {\tt A} & {\tt C} & {\tt C} & \to \\
     \gets & \rotatebox[origin=c]{180}{\tt U} & \rotatebox[origin=c]{180}{\tt G} & \rotatebox[origin=c]{180}{\tt G} & {-}
\end{array}$; then an RNA sequence on which these sequences occur as non-overlapping factors can fold into a hairpin-like structure as {\tt \underline{ACC}UUC\underline{GGU}}. 
A \textit{hairpin} can be modeled as a word of the form $x\ell \theta(x)$ for some $x, \ell \in \Sigma^*$, where $x$ and $\theta(x)$ are to hybridize with each other into a \textit{stem} and leaves $\ell$ as a \textit{loop}. 
By letting both $x$ and $\ell$ be empty, this definition enables us to handle the scenario that hairpin formation in co-transcriptional splicing is nothing but a by-product of splice-site recognitions in a unified framework proposed below. 
In reality, neither of them can be empty, and in particular, $\ell$ must be of length at least 3 due to the stiffness of RNA sequences \cite{KariLKST06}. 
Let 
\begin{equation}\label{eq:hairpin_set}
    H_{\Sigma, \theta} = \{x\ell \theta(x) \mid x, \ell \in \Sigma^*\}
\end{equation} 
be the set of all hairpins; the subscript is omitted whenever $\Sigma$ and $\theta$ are clear from the context. 

Unstable structures in general may form, but cannot stay for a long time. 
The stability of a hairpin is rewarded by a longer stem (more base pairs) while it is penalized by a longer loop (more unbound bases). 
It is well known that contribution by a stem of $n$ base pairs is a sum of \textit{stacking} contributions by the $n{-}1$ interfaces between neighboring bases pairs in the stem, which have been experimentally calculated for all possible combinations of two base pairs \cite{FreierKJSCNT86}. 
Destabilization by a loop is known to depend only on the loop length, but little is known about how the penalty depends on the length. 
Tetraloops, that is, loops of length 4 are prevailing in nature, but at the same time, hairpins with a considerably longer loop than a stem are observed therein. 
Therefore, we study the following three energy models: \textit{logarithmic loop penalty}, \textit{linear loop penalty}, and \textit{constantly-bounded loop-length}. 
The last one is in fact not an energy model, but rather aims at modeling greediness of RNA co-transcriptionality, which may prevent a factor of an RNA sequence from waiting for hybridization for a long time. 
It is simply parameterized by a single constant $d$ which bounds the loop length of a hairpin to be considered valid. 
For the sake of arguments below, it shall be considered as an energy model, and any hairpin whose loop is of length at most the constant shall be rather said stable. 
In the other two models, we rather say that a hairpin $x \ell \theta(x)$ is \textit{stable} if $d_1|\ell| -d_2|x| \le d$ (linear penalty) or $d_1 \log(|\ell|) - d_2|x| \le d$ (logarithmic penalty) for some constants $d_1, d_2, d$, which should be determined experimentally. 
By $H_{{\rm con}, (c)}$, $H_{{\rm lin}, (d_1, d_2, d)}$, and $H_{{\rm log}, (d_1, d_2, d)}$, we denote the respective sets of stable hairpins in the constantly-bounded loop-length, linear loop penalty, and logarithmic loop penalty models. 
Both the specification of a model and that of parameters in the subscript may be omitted. 

\vspace{2mm}
\begin{example}
    Assuming $\theta$ defined as the Watson-Crick complementarity from above, the word $u_1 = \mathtt{\underline{UUAGGA}GUAA\underline{UCCUAA}}$ can be considered as a hairpin with a tetra-loop $\mathtt{CUUG}$. Regarding the energy models defined above, setting a constant $d = 4$ would make $u_1$ a valid hairpin under the constantly-bounded loop-length model. For the linear penalty model, if we set $d_1 = d_2 = 1$ and $d = 0$, under which, for a stable hairpin to form, the stem must be at least as long as the loop, then $u_1$ is also a stable hairpin, even if we cut the stem and are only left with $\mathtt{\underline{AGGA}GUAA\underline{UCCU}}$. For the logarithmic penalty model, under the same parameters, we could even cut the stem even further, e.g., to $\mathtt{\underline{GA}GUAA\underline{UC}}$, and still obtain a stable hairpin.
\end{example}
\vspace{2mm}

Having defined hairpins and the three energy models for their stability, now we are ready to formally model RNA co-transcriptional splicing. 
As shown in Figure~\ref{fig:splicing}, this phenomenon is a variant of contextual hairpin deletion, in which a short left context $\alpha$ called the \textit{5' splice site (5'SS)} is first recognized and kept close to the RNA polymerase along with a subsequent factor $x$ until the complement $\theta(\alpha x)$ is transcribed and then hybridized, yielding a hairpin of the form $\alpha x \ell \theta(\alpha x)$ for some $\ell$; this hairpin is spliced out along with a short subsequent factor that ends with a right context called the \textit{3' splice site (3'SS)}; the structure spliced out looks more like a lariat. 

Let us model this contextual deletion of a lariat. 
Given an alphabet $\Sigma$, an antimorphic involution $\theta: \Sigma^* \to \Sigma^*$, and a finite set $C \subseteq \Sigma^* \times \Sigma^*$ of pairs of words (\textit{contexts}), this deletion can be modeled as an operation to obtain $uv$ from $u \alpha x \ell \theta(\alpha x) z \beta v$ for some $(\alpha, \beta) \in C$ and $u, x, \ell, z, v \in \Sigma^*$ such that $\alpha x \ell \theta(\alpha x)$ is a stable hairpin in an energy model to be considered in the context, and $|z| \le g$ for some bound $g$ on the length of the \textit{gap} between the end of the hairpin $\alpha x \ell \theta(\alpha x)$ (called branch point; see Figure~\ref{fig:splicing}) and the right context $\beta$. 
Then we write $u \alpha x \ell \theta(\alpha x) z \beta v \hpdbe_{\lbrack C \rbrack} uv$. 
In order to clearly state which energy model is considered under what parameters, we augment this notation as: $\hpdblog_{\lbrack C \rbrack, (d_1, d_2, d)}$, $\hpdblin_{\lbrack C \rbrack, (d_1, d_2, d)}$, and $\hpdbcon_{\lbrack C \rbrack, (d)}$. 
Let us call this operation \textit{contextual lariat deletion}. 
Note that we use the term ``context'' interchangeably for an element of $C$ and for the set itself, as long as no confusion arises. 

\vspace{2mm}
\begin{example}\label{example:contextual-lariat-deletion}
    Assume $\theta$ defined as the Watson-Crick complementarity from above. 
    Let $C = \{({\color{red}\tA\tA},{\color{blue}\tG\tG\tG})\}$ be a context-set with only one context. Consider the word 
    $$w = \mathtt{CUC\underline{{\color{red}AA}A}CGGC{\underline{UUU}}{\color{purple}CC}{\color{blue}GGG}C}.$$
    Then, assuming the gap-length bound $\delta$ is big enough such that $|z| \leq \delta$, we can factorize $w$ into $u = \tC\tU\tC$, $\alpha = \underline{{\color{red}\tA\tA}}$, $x = \underline{\tA}$, $\ell = \tC\tG\tG\tC$, $\theta(\alpha x) = \underline{\tU\tU\tU}$, $z = {\color{purple}\tC\tC}$, $\beta = {\color{blue}\tG\tG\tG}$, and $v = \tC$ and, using the context $({\color{red}\tA\tA},{\color{blue}\tG\tG\tG})$, obtain the following derivation by contextual lariat deletion:
    $$w = u\alpha x\ell\theta(\alpha x)z\beta v = \mathtt{CUC\underline{{\color{red}AA}A}CGGC{\underline{UUU}}{\color{purple}CC}{\color{blue}GGG}C} \hpdbe_{[C]} \tC\tU\tC\tC = uv $$
    Notice, that under the constantly-bounded energy model, $d\geq 4$ must be set for $\alpha x\ell\theta(\alpha x)$ to form a stable hairpin, as $|\ell| = 4$. Assuming the parameters $d_1 = d_2 = 1$ and $d = 0$, notice that $\alpha x\ell\theta(\alpha x)$ only forms a stable hairpin under the logarithmic energy model, not the linear one, as the stem ($|\alpha x|$) is shorter than the loop ($|\ell|$).
\end{example}
\vspace{2mm}

At this primitive stage of theoretical research, we have not been provided with enough evidence to dismiss the possibility that co-transcriptional splicing merely requires to recognize the 5' and 3' splice sites and hairpin formation is just a by-product of these recognitions. 
In this scenario, 1-step co-transcriptional splicing should be formulated rather as an operation to remove a factor from a word on condition that the factor is of the form $\alpha w \beta$ for some $(\alpha, \beta) \in C$ as $u \alpha w \beta v \to uv$. 
Unlike in most variants of contextual deletion \cite{ChoHKS18, KariT96}, contexts ($\alpha$ and $\beta$ here) do not remain. 
As Kari and Thierrin \cite{KariT96} used the notation $\xrightarrow[(\alpha, \beta)]{}$ for their contextual deletion as $u \alpha w \beta v \xrightarrow[(\alpha, \beta)]{} u\alpha \beta v$, we rather employ $\xrightarrow[\lbrack \alpha, \beta \rbrack]{}$ (actually $\to_{[\alpha, \beta]}$ to save space) to emphasize that the factor to be removed spans up to $\alpha$ leftward and up to $\beta$ rightward, that is, $u \alpha w \beta v \to_{[\alpha, \beta]} uv$. 
We can naturally extend this operation for a finite set $C$ of contexts as $u \alpha w \beta v \to_{[C]} uv$ if $(\alpha, \beta) \in C$. 
Then we call this operation \textit{bracketed contextual deletion}. 

These two operations should be applied greedily as well as iterated or parallelized for more realistic models of co-transcrptional splicing. 
There are left-greediness and right-greediness. 
The derivation $u\alpha w \beta v \to_{[C]} uv$ is \textit{left-greedy} if for any proper prefix $u_p$ of $u$ and a context $(\alpha', \beta') \in C$, if $u_p \alpha'$ is a prefix of $u\alpha w \beta v$ as $u_p \alpha' x = u\alpha w \beta v$, then $\beta'$ does not occur in $x$. 
It is rather said \textit{right-greedy} if $\beta$ is not a proper factor of $w\beta$, that is, $\beta$ at the end of $w\beta$ is the first occurrence of $\beta$ after $\alpha$ is fully read. 
Both definitions can be adapted to the context of contextual lariat deletion. 
Being applied left and right-greedily, the respective brackets are replaced by double square brackets like $\to_{[[C]}$ and $\to_{[[C]}$. 
The definition of left-greediness should be controversial particularly in the context of co-transcriptional splicing; encountering at $\alpha'$, how can co-transcriptional folding know that the matching $\beta'$ never appears?
From the perspective of molecular programming, we should engineer a DNA template that is free from such a pseudo-context. 
In this paper, hence we consider only the right-greediness. 

\vspace{2mm}
\begin{example}
    Let $C = \{({\color{red}\tA\tC\tA},{\color{blue}\tG\tG\tU}),({\color{purple}\tC\tC\tG},{\color{orange}\tA\tU\tA})\}$ be some context-set. Let
    $$w = \mathtt{A {\color{purple}CCG} AA {\color{red}ACA} AA {\color{blue}GGU} AAA {\color{blue}GGU} AA}.$$
    Then, with bracketed contextual deletion, we could remove the factor $\mathtt{{\color{red}ACA} AA {\color{blue}GGU}}$ or the factor $\mathtt{{\color{red}ACA} AA {\color{blue}GGU} AAA {\color{blue}GGU}}$. With right greedy bracketed contextual deletion, no fitting right context can be a proper factor of the removed part, hence, here we can only remove the factor $\mathtt{{\color{red}ACA} AA {\color{blue}GGU}}$ (using the first occurrence of the right context after a choosing a left context). Notice, that there exists no right context ${\color{orange}\tA\tU\tA}$ for the left context ${\color{purple}\tC\tC\tG}$ in $w$. A left-greedy perspective might assume that we are stuck to ${\color{purple}\tC\tC\tG}$ and cannot remove anything else, which could be unpractical.
\end{example}
\vspace{2mm}

The iterated extensions of the previously introduced operations can be modeled as their reflexive and transitive closure, as is done conventionally, and denoted by $\hpdbe_{\lbrack C \rbrack}^*$ and $\to_{[C]}^*$, respectively. 
As for non-overlapping, parallel, extensions, the contextual lariat deletion is parallelized recursively as 
\[
    w \ {\hpdbp_{\lbrack C \rbrack}} \ w'
\]
if either $w \hpdbe_{\lbrack C \rbrack} w'$ or $w$ and $w'$ can be factorized as $w = uv$ and $w' = u'v'$ for some $u, v, u', v' \in \Sigma^*$ such that $u \hpdbe_{\lbrack C \rbrack} u'$ and $v \ {\hpdbp_{\lbrack C \rbrack}} \ v'$. A single application of contextual lariat deletion, as defined above, is also called \emph{1-step} deletion in this context, denoted by $w\hpdboe_{[C]}w'$. Notice that 1-step deletion is the most basic case of parallel deletion.  Later on, this will be important in the context of the sets of all words obtainable by a single operation type. If we consider the parallel or 1-step variants of contextual lariat deletion in the context of different energy models, instead of writing just $\mathtt{p}$ or $1$, we write $\mathtt{p'X}$ (resp. $1'X$), for $X\in\{\mathtt{lin},\mathtt{log},\mathtt{con}\}$, above the operator.
The bracketed contextual deletion can be parallelized likewise.

\vspace{2mm}
\begin{example}
    Assume $\theta$ defined as the Watson-Crick complementarity from above. Let $C = \{({\color{red}\tA\tA},{\color{blue}\tG\tG\tG})\}$, again, be a context-set with only one context.
    Consider an extended version of the word considered in Example \ref{example:contextual-lariat-deletion}:
    $$w = \mathtt{CUC\underline{{\color{red}AA}A}CGGC{\underline{UUU}}{\color{purple}CC}{\color{blue}GGG}CC\underline{{\color{red}AA}UAU}CUUC\underline{AUAUU}{\color{purple}C}{\color{blue}GGG}C}.$$ 
    Applying parallel contextual lariat deletion, we can remove two factors at once (factorizations marked by colors and underline as in the previous examples) and, assuming the gap-length bound $\delta$ to be larger or equal to $|z|$ for both removed factors, we obtain
    $$ w = \mathtt{CUC\underline{{\color{red}AA}A}CGGC{\underline{UUU}}{\color{purple}CC}{\color{blue}GGG}CC\underline{{\color{red}AA}UAU}CUUC\underline{AUAUU}{\color{purple}C}{\color{blue}GGG}C} \hpdbpe_{\lbrack C \rbrack} \tC\tU\tC\tC\tC\tC.$$
    Notice, that a constant loop-length bound of $d \geq 4$ still suffices to remove both lariat-structures. If we assume the parameters $d_1 = d_2 = 1$ and $d = 0$, then contextual lariat deletion under the linear energy model only allows for the second factor to be removed. Under the logarithmic model, however, we could still remove both factors at once.
\end{example}
\vspace{2mm}

Let ${\rm OP}$ be one of the operations defined so far. 
Then, given a word $w$, we denote by $[w]_{\rm OP}$ the set of all the words that can be obtained by applying the operation {\rm OP} to $w$. 
For a language $L$, let $[L]_{\rm OP} = \bigcup_{w \in L} [w]_{\rm OP}$.

%% file: sec-3-new-temp-cons.tex
This section formalizes the problems of deciding whether a single DNA template can be engineered so that co-transcriptional splicing yields a prescribed family of RNA words. We begin (3.1) with bracketed contextual deletion and introduce three decision problems: verification, exact constructability, and the more practical length-bounded constructability. Polynomial-time algorithms are given for verification and for the exact variant, whereas the length-bounded problem is proved NP-complete over alphabets with at least four symbols.
In 3.2 we repeat the analysis for the biologically motivated contextual lariat deletion model (constant, linear, logarithmic loops). Here, verification remains in P, but the general constructability task is again NP-complete in the parallel setting.

First we consider the problem of encoding a finite set of RNA sequences into a single DNA template (or a limited number of templates, though not being considered in this paper) from which all the RNA sequences can be retrieved by co-transcriptional splicing. 
This problem can be formalized in the most uncompromising way as follows: 

\vspace{2mm}
\begin{problem}[Exact Template Constructability]\label{problem:construct-exact-match}
    Given a finite set of target (RNA) words $R = \{w_1, ..., w_n\} \subseteq \Sigma^*$ and a finite set $C$ of contexts, does there exist a (template) word $t \in\Sigma^*$ such that the set of all words obtainable from $t$ by co-transcriptional splicing is exactly $R$ in a supposed model?
\end{problem}
\vspace{2mm}

This formalization can, of course, allow for contextual lariat deletion instead. 
However, as observed in the following, this setting is too restrictive for arbitrary targets to be accommodated in a single template, as co-transcriptional splicing is intrinsically length-decreasing. 

\vspace{2mm}
\begin{observation}\label{lem:exactsetfromword}
    A solution $t$ to Problem~\ref{problem:construct-exact-match} must satisfy all the following properties: 
    \begin{enumerate}
        \item such $t$ is unique;
        \item $t \in R$ and all the other words in $R$ are strictly shorter than $t$.
    \end{enumerate}
\end{observation}
\vspace{2mm}

Tolerance in nature and experiments encourages the formalization to be relaxed. 
Biological processes are inherently stochastic, and RNA is naturally degraded by RNase. 
These target RNA sequences are to interact with each other via designated complementary factors called \textit{domains}. 
Thus, unless competing with target sequences for these domains, any byproduct should be tolerated. 
It then suffices to concatenate (finite number of) target sequences as $t = xyw_1xyw_2xy \cdots xyw_nxy$ via $xy$'s using a single context $(x, y)$; from this $t$ we can retrieve all the targets by iterated bracketed contextual deletion $\to_{[x, y]}^*$, along with non-targets like $w_1w_2$ and $xyw_2$. 
This trivial but costly solution motivates the length bound $k$ in the following formalization. 

\vspace{2mm}
\begin{problem}[Template Constructability]\label{problem:main-construction}
    Given a finite set of words $R = \{w_1,...,w_n\} \subseteq \Sigma^*$, a finite set of contexts $C$, and an integer $k \ge 1$, does there exist a word $t\in\Sigma^*$ of length at most $k$ such that all the words in $R$ can be obtained from $t$ by co-transcriptional splicing in a supposed model?
\end{problem}
\vspace{2mm}

Let us also formalize its verification variant. 

\vspace{2mm}
\begin{problem}[Template Verification]\label{problem:verification}
    Given a finite set of words $R = \{w_1,...,w_n\} \subseteq \Sigma^*$, a finite set of contexts $C$, and a word $t \in\Sigma^*$, can all the words in $R$ be obtained from $t$ by co-transcriptional splicing in a supposed model?
\end{problem}
\vspace{2mm}

\subsection{On the bracketed contextual deletion}

 First we argue that the Exact Template Constructability problem can be decided efficiently due to the fact the set $R$ in the positive instances must contain a single largest word and all other elements of the transcription set of this largest word with respect to some context-set $C$. Formally, from Observation~\ref{lem:exactsetfromword}, we get that Problem~\ref{problem:construct-exact-match} has a solution if and only if all of the following hold: (1) there is a single longest word $w_{max}$ in $R$, (2) the closure $[w_{max}]_{\hpdubi_C}$ contains $R$ and (3) $[w]_{\hpdubi_C}\setminus R=\emptyset$. Conditions (1) and (3) can be checked efficiently using standard polynomial time algorithms found in most textbooks, whereas condition (2) is equivalent to Problem~\ref{problem:verification} (Template Verification) that we will tackle after, and here we use as a subroutine. As we will see in what follows, the complexity of Template Verification dominates the complexity of checking conditions (1) and (3), so please see the complexity discussion there.

In summary, Problem~\ref{problem:construct-exact-match} is decidable in polynomial time (implicitly yielding also the template $w$ if it exists). By Remark \ref{remark:right_greedy} we have that an analogous construction works also for the greedy version of Problem \ref{problem:construct-exact-match} yielding efficient decidability there as well.

Next we show that template verification can be decided in the iterated bracketed contextual deletion model in polynomial time by effectively constructing a DFA recognizing $[w]_{\to_{[C]}^*}$ for a given word $w$, using standard algorithmic properties of regular languages.

\vspace{2mm}
\begin{remark}
    It does not suffice to explicitly construct each word in $[w]_{\to_{[C]}}^*$, as this set might contain an exponential number of words in $|w|$, $|C|$ and $|\Sigma|$. Consider the word $w = a_1...a_n$, for some alphabet $\Sigma_n = \{a_1, a_2, \ldots, a_n\}$, and large $n\in\N$. 
    The context $C = \{\ (a_i,a_j)\mid a_i,a_j\in\Sigma_n, i < j\ \}$ is of size $\sum_{i\in[n]}(i)$, which is polynomial in $n$, but $[w]_{\to_{[C]}}^*$ is the set of all subsequences of $w$ where each missing factor has at least length $2$.
\end{remark}

        \alg{algo:inverse}\paragraph*{Algorithm~\ref{algo:inverse}: construct an NFA for $[w]_{\to_{[C]}^*}$ from a given $w$}
        \begin{enumerate}
            \item[0.] Construct DFA $A=(Q,\Sigma,q_0,\delta,F)$ that accepts the language $\{w\}$.
            \item[1.] Set $changed=\mathrm{False}$
            \item[2.] For each pair $(p,q)\in Q\times Q$ with $p\neq q$ and for each context $(x,y)\in C$:
            \begin{enumerate}
                \item[3.] If there is a word $xzy$ such that $q\in \delta(p, xzy)$, add an $\varepsilon$-transition between $p$ and $q$ and remove $(p,q)$ from $Q$, set $changed=\mathrm{True}$.
                
            \end{enumerate}
            \item[4.] Remove the $\varepsilon$-transitions by the $\varepsilon$-closure method.
            \item[5.] If $Q$ is empty or $changed=\mathrm{False}$, stop. Otherwise, continue from 1.
        \end{enumerate}

    \begin{example}
        Let $w=abcd$ and $C=\{(a,c),(c,d)\}$. 
        The minimal DFA accepting only $w$ is $(\{q_0,\dots,q_4\},\{a,b,c,d\},q_0,\delta,\{q_3\})$ where $\delta(q_i,w[i+1])=q_{i+1}$ for all $i\in [0,2]$. After setting $changed=\mathrm{False}$, we enter the loop in 2., and find the pair of states $p=q_0, q=q_3$, context $(x,y)$ and word $xzy$ with $x=a$, $z=b$, $y=c$ such that $\delta(p,xzy)=\delta(q_0,abc)=q_3$, so we add an $\varepsilon$-transition between $q_0$ and $q_3$. We do the same for $q_2$ and $q_4$ due to the context $(c,d)$. We set $changed=\mathrm{True}$ and move to step $4.$ Here we remove the $\varepsilon$-transitions by the classical method of constructing $\varepsilon$-closures for each state. This results in adding a transition $q_4\in\delta(q_0,d)$ and adding $q_2$ to the set of final states, since its $\varepsilon$-closure includes $q_4$. The new automaton, on top of $abcd$, now also accepts the words $ab$ and $d$, consistent with the hairpin deletion of the factors $cd$ and $abc$, respectively. In the next round, $changed$ is set to $\mathrm{False}$ again and this time the loop in $2.$ does not modify the value, so the algorithm finishes after step 5.
    \end{example}

\vspace{2mm}
\begin{proposition}\label{prop:problem-1-P}
    Algorithm~\ref{algo:inverse} constructs an NFA without $\varepsilon$-transitions accepting $[w]_{\to_{[C]}^*}$. 
    As a consequence, template verification with respect to the iterated bracketed contextual deletion is decidable in $\mathcal{O}(|w|^7\cdot |C|^2 + |w|^2\cdot \sum_{u\in W}|u|)$ time.
\end{proposition}
    \begin{proof}
        Constructing the DFA $A$ accepting $w$ is trivial and takes $\mathcal{O}(|w|)$ time, so let $A=(\{q_0,\dots,q_m\},\Sigma, q_0, \delta , \{q_m\})$, where $m=|w|$, such that $\delta(q_{i-1},a_i)=q_i$ for all $i\in [m]$, where $w=a_1\cdots a_m$. This machine recognizes $\{w\}$, and it will be modified to obtain a NFA for $[w]_{\to_{[C]}^*}$. Let $Q_2=\{(q_i,q_j) \mid i\neq j\}$. 
    \noindent
    To execute Algorithm~\ref{algo:inverse}, first, for each $(x,y)\in C$ the DFA $M_{x,y}$, accepting $x\Sigma^*y$ is constructed. This can be done in $\bigo(|xy|^2)$ time and results in $M_{x,y}$ having $|xy|+1$ states (the sink state is omitted). 
    The whole loop in the algorithm is executed at most $|Q_2|\cdot |C| \in \bigo(m^2|C|)$ times. In each cycle and for each pair of $(p,q)\in Q_2$ and $(x,y)\in C$, step 3. is executed by checking whether the current NFA $A'$ accepts any word in $x\Sigma^*y$ starting from state $p$ and ending at $q$. As $A'$ has the same $|Q|$ number of states as $A$, and $\bigo(|Q|^2)$ transitions, this can be done in time $\bigo(|Q|^2\cdot |xy|)=\bigo(m^2\cdot |xy|)$ by checking $L(A')\cap L(M_{x,y})=\emptyset$ using the product automaton of $A'$ and $M_{x,y}$. This means that step 2. finishes in time $$\bigo(m^2\cdot C\cdot m^2\cdot |xy|)=\bigo(m^4C\cdot |xy|)\subseteq\bigo(m^5C),$$ 
    where the last inclusion follows from the fact that for contexts $(x,y)$ such that $|xy|>|w|$ we do not need to perform any checks, as the context pair cannot occur without overlap in $w$.
    
    Removing the $\varepsilon$ transitions from $A'$ in step 4. takes at most $\bigo(m^4\cdot |\Sigma|)$ time as $\bigo(m^2)$ is the maximum number of $\varepsilon$ transitions and $\bigo(m^2\cdot |\Sigma|)$ is the number of other transitions. The complexity of this step is dominated by that of step 2, so we get that the algorithm construct the NFA $A'$ accepting $[w]_{\to_{[C]}^*}$ in time $\bigo(m^2|C|\cdot m^5|C|)=\bigo(m^7\cdot|C|^2)$.
        
    After running Algorithm~\ref{algo:inverse}, an NFA with $m+1$ states without $\varepsilon$-transitions is obtained that accepts the language $[w]_{\to_{[C]}^*}$. Whether a word $w_i$ is in the language can be decided in time $\bigo(m^2|w_i|)$. Decide that for each $w_i$, in turn, and answer yes, if each of them is in the language.
\end{proof}

\begin{remark}\label{remark:right_greedy}
    Notice that Proposition \ref{prop:problem-1-P} can be adapted for the use of right-greediness.
    It suffices to add in step 2 of Algorithm \ref{algo:inverse} the constraint that $q\in\delta(p,xzy)$ for $y$ not being a factor of $zy[1...|y|-1]$. This increases the complexity bound somewhat, as now we are checking $L(A')\cap x\Sigma^* y\cap \overline{x\Sigma^* y\Sigma^+}$ for emptiness (where the last term in the intersection filters out words that have an occurrence of $y$ prior to the suffix $y$). The constructed product NFA for that step still has size and construction time that is polynomial in the size of the given instance of template verification. 
\end{remark}
\vspace{2mm}

Unlike the verification, template constructability turns out to be NP-complete, at least over an alphabet of size at least 4. 
The proof is by a reduction from the decision variant of the Shortest Common Supersequence problem (SCSe), which is known to be NP-complete for binary or larger alphabets \cite{Middendorf1994}.

\vspace{2mm}
\begin{remark}
    Problem \ref{problem:main-construction} is defined for the non-greedy version of bracketed contextual deletion. Notice, that all following results also work with greedy iterated bracketed contextual deletion, resulting in NP-completeness of such a problem variant as well.
\end{remark}
\vspace{2mm}

\begin{theorem}\label{thm:main-prob-decision-hard}
    Template Constructability with respect to the iterated bracketed contextual deletion is NP-complete for $|\Sigma| \geq 4$.
\end{theorem}
\begin{proof}
     First we give an argument for the containment in NP. As a witness for a yes instance of the problem, it needs only to contain the template word $w$ with length $\leq k$ (as mentioned before, $k$ is linearly upper bounded by the size of the words in $R$ and the size of contexts in $C$, as otherwise trivial solutions always exist), a series of contexts from $C$, and the indices at which the bracketed contextual deletion steps occur that yield each $w_i\in R$. The length of the series of contexts and indices is linear in $|w|$, as each splicing removes at least one symbol. \\

    Let us now turn our attention to the hardness. 
    Let $(W,k)$ be an instance of SCSe such that $W\subset\Sigma^*$ is a finite set of words over some alphabet $\Sigma$ with $|\Sigma|\geq 2$ and with $k\in\N$. SCSe asks whether there exists a common supersequence of $W$ with length at most $k$. Encode the instance $(W,k)$ to an instance $(R',C,k')$ of Problem \ref{problem:main-construction} in the following way. Let $W = \{w_1, w_2, ... , w_n\}$ for some words $w_i\in\Sigma^*$ and $i\in[n]$. Then we construct $R' = \{w_1', w_2',...,w_n'\}$ by setting 
    $$w_i' = \#_s w_i[1] \#_e\ \#_s w_i[2] \#_e\ ... \#_s w_i[|w_i|] \#_e$$
    for some new letters $\#_s,\#_e\notin\Sigma$. Next, set $C = \{(\#_s, \#_e)\}$. Finally, let $k' = 3k$.
    
    First, assume there exists a solution $w\in\Sigma^*$ for the SCSe instance $(W,k)$ such that $|w| \leq k$ and $W \subseteq \subseq(w)$. Then, construct a word $w' \in (\Sigma \cup \{\#_s,\#_e\})^*$ such that
    $$ w' = \#_s w[1] \#_e\ \#_s w[2] \#_e\ ...\ \#_s w[|w|] \#_e. $$
    Then $|w'| = 3|w| \leq 3k$. Notice that $w\in\subseq(w')$ and $w_i\in\subseq(w)$, so $w_i\in\subseq(w')$. As $w_i\in\subseq(w')$, there exists a sequence $(\ell_1,\ell_2,...,\ell_{|w_i|})\in[|w'|]^{|w_i|}$ of indices in $w'$ such that
    $$w_i = w'[\ell_1]w'[\ell_2]...w'[\ell_{|w_i|}].$$
    By construction of $w'$, every occurrence of a letter of $\Sigma$ in $w'$ is surrounded by $\#_s$ and $\#_e$. So, $w_i' = \#_s w'[\ell_1] \#_e \#_s w'[\ell_2] \#_e ... \#_s w'[\ell_{|w_i|}]\#_e$ is a subsequence of $w'$. Also, by construction, notice that all other factors can be removed by bracketed contextual deletion, as their structure can be split up in factors $\#_s a \#_e$ for $a\in\Sigma$ and $(\#_s,\#_e)\in C$. Hence, $w' \to_{[C]}^* w_i'$ is obtained and by that it can be seen that $R'\subseteq [w']_{\to_{[C]}^*}$.
    
    Now, assume there exists a solution word $w'\in(\Sigma\cup\{\#_s,\#_e\})^*$ for the constructed instance $(R',C,3k)$ with $|w'|\leq 3k$ such that $R'\subseteq [w']_{\to_{[C]}^*}$. Begin this direction of the proof by showing that one can assume $w'$ to be of a very specific structure.
    
    (1) Suppose $w'$ contains a factor $a_ia_j\in\Sigma^2$, so $w' = u_1 a_ia_j u_2$ for some $u_1,u_2\in(\Sigma\cup\{\#_s,\#_e\})^*$. No word in $R'$ contains any factor of two subsequent letters from $\Sigma$. Also, no splicing context in $C$ contains any letter of $\Sigma$. By that, since $w'\to_{[C]}^* w_i'$ for all $w_i'\in R'$, also $u_1u_2 \to_{[C]}^* w_i'$ holds for all $w_i'\in R'$. Notice that $|u_1u_2| < |w'| \leq 3k$. Hence, from now on assume w.l.o.g. that $w'$ does not contain such a factor and by that
    $$w' = u_1\ a_{j_1}\ u_2\ a_{j_2}\ u_3\ ...\ u_m\ a_{j_m}\ u_{m+1}$$
    for some $a_{j_1},...,a_{j_m}\in\Sigma$, $u_2,...,u_m,\in\{\#_s,\#_e\}^+$, and $u_1,u_{m+1}\in\{\#_s,\#_e\}^*$.
    
    (2) Next, suppose $u_1 = \varepsilon$ ($u_{m+1} = \varepsilon$). Then no word from $R'$ is obtainable by bracketed contextual deletion as no word in $R'$ starts (ends) with a letter from $\Sigma$ and $a_{j_1}$ ($a_{j_m}$) cannot be removed by bracketed contextual deletion with the given $C=\{\#_s,\#_e\}$. So, $|u_1| \geq 1$ and $|u_{m+1}| \geq 1$.
    
    (3) Now, suppose $|u_i| = 1$ for some $i\in[m]\setminus{\{1\}}$. Then $u_i = \#_s$ or $u_i = \#_e$.
        First, assume $u_i = \#_s$. Then one has 
        $$w' = v_1\ u_{i-1}\ a_{j_{i-1}}\ \#_s\ a_{j_i}\ u_{i+1}\ v_2$$
        for some $v_1,v_2\in(\Sigma\cup\{\#_s,\#_e\})^*$ and $u_{i-1},u_{i+1}\in\{\#_s,\#_e\}^+$. Notice that if $u_{i+1}v_2$ contains some occurrence of $\#_e$, then we could remove $\#_sa_{j_i}v_3$ by bracketed contextual deletion, for some $v_3\in(\Sigma\cup\{\#_s,\#_e\})$, and keep $a_{j_{i-1}}$. However, as $u_i$ does not contain an occurrence of $\#_e$,
        one cannot remove $a_{j_{i-1}}\#_s$ by bracketed contextual deletion to only keep that occurrence of $a_{j_i}$. Also, $a_{j_{i-1}}\#_sa_{j_i}$ does not occur in any word from $R'$ as a factor. So, the only possible scenarios to obtain some word in $W'$ are that the occurrence of $a_{j_{i-1}}$ could be kept after deletion, removing $\#_sa_{j_i}$, or that the whole factor $a_{j_{i-1}}\#_s a_{j_i}$ is removed by bracketed contextual deletion. Hence, if for some $w_i'\in R'$ one has
        $$w' = v_1u_{i-1}\ a_{j_{i-1}}\#_sa_{j_i}\ u_{i+1}v_2 \to_{[C]}^* w_i',$$
        then also
        $$ v_1u_{i-1}\ a_{j_{i-1}}\#_s\ u_{i+1}v_2 \to_{[C]}^* w_i',$$
        i.e., we can just delete that occurrence of $a_{j_i}$, as it is never used in any obtained $w_i'$. This effectively merges $u_i$ with $u_{i+1}$, resulting in a new factor $u_k\in\{\#_s,\#_e\}^+$ of at least length $2$.
        The same can be obtained by a symmetric argument if $u_i = \#_e$, just in the other direction, i.e., if
        $$w' = v_1u_{i-1}\ a_{j_{i-1}}\#_e a_{j_i}\ u_{i+1}v_2 \to_{[C]}^* w_i',$$
        then also
        $$ v_1u_{i-1}\ \#_ea_{j_i}\ u_{i+1}v_2 \to_{[C]}^* w_i'.$$
        Most importantly, see that the size of the alternative template is smaller than the one of $w'$.
        
        (4) By the arguments in (1), (2), and (3) it can be concluded that if there exists a word $w'\in(\Sigma\cup\{\#_s,\#_e\})$ with $|w'|\leq 3k$ and $R'\subseteq [w']_{\hpdubi_C}$, then there also exists a word $w''\in(\Sigma\cup\{\#_s,\#_e\})$ with $|w''| \leq |w'| \leq 3k$, i.e., of smaller or equal size to $w'$, for which we also have $R'\subseteq [w'']_{\to_{[C]}^*}$ and which has the structure
        $$w'' = u_1\ a_{k_1}\ u_2\ a_{k_2}\ u_3\ ...\ u_{m'}\ a_{k_{m'}}\ u_{m'+1}$$
        for $a_{k_1},...,a_{k_{m'}}\in\Sigma$, $u_1,u_{m'+1}\in\{\#_s,\#_e\}^+$, and $u_2,...,u_{m'}\in\{\#_s,\#_e\}^{\geq 2}$. 
        
        Using $w''$ of (4), the existence of a supersequence $w\in\Sigma^*$ of $W$ of at most length $k$ is now shown. Let $w = a_{k_1}a_{k_2}...a_{k_{m'}}$. Then $|w| \leq k$ as $|w''| \leq 3k$ and $|u_1u_2...u_{m'+1}| \geq 2|w|$ by the previous assumptions on the structure of $w''$. Consider an arbitrary word $w_i\in W$. By (4), one knows that $w'' \to_{[C]}^* w_i'$. So, $w_i'$ is a subsequence of $w''$. As $w_i$ is a subsequence of $w_i'$ by construction, $w_i$ is also a subsequence of $w''$. By that, as $w$ is the sequence of all letters in $w''$ which are from $\Sigma$ and $w_i$ only contains letters from $\Sigma$, it follows that $w_i$ has to appear in $w$ as a subsequence. Hence, it can be seen that $w$ is a supersequence of $W$ with $|w| \leq k$. This concludes both directions of the reduction. 
\end{proof}

This concludes the results regarding the template constructibility problem for bracketed contextual deletion. Notice that the general variant of the problem is NP-hard, but in restricted cases, efficient ways to construct such a template may still be obtained. Identifying those and providing specific constructions is to be done in future research. In the following final part of this section, we extend the view of the problems of this subsection to contextual lariat deletion.

\subsection{On the contextual lariat deletion}

We now consider the template verification and constructability with respect to contextual lariat deletion. Actually, this time, we consider parallel deletion instead of iterated deletion, as it is more suitable for this more realistic model. We do not explicitly consider the exact variant of the Template Constructability Problem as the observed restrictions for valid instances in the bracketed contextual deletion model also follow in a similar manner for contextual lariat deletion, resulting in a very restrictive setting here as well.

Let us begin with verification. 
As it turns out, extending the idea from Algorithm~\ref{algo:inverse} and Proposition \ref{prop:problem-1-P}, it can be shown that templates can be verified efficiently even with respect to contextual lariat deletion under all energy models.

\vspace{2mm}
\begin{proposition}\label{prop:hpa-verification-in-p}
    Template verification with respect to parallel as well as 1-step contextual lariat deletion is in $P$ under all three energy models. Moreover, an NFA accepting $[w]_{\hpdbp_{[C]}}$ (resp. $[w]_{\hpdbo_{[C]}}$) can be constructed in time polynomial in $|w|$ and $|C|$ for some given word $w\in\Sigma^*$, a set of contexts $C$, and energy parameters.
\end{proposition}
\begin{proof}
    Start with the same NFA $A = (Q,\Sigma,q_0,\delta,F)$ with $\{q_1,...,q_{|w|}\}$ and $F = \{q_{|w|}\}$ as in Proposition \ref{prop:problem-1-P}. Hence, $A$ is representing a line of states that may only read $w$. Similar to the proof given in Proposition~\ref{prop:problem-1-P}, for each pair $(q_i,q_j)\in Q$ for $i < j$ and for each splicing context $(\alpha,\beta)\in C$, we check whether we have $w[i..j] = \alpha x \ell \theta(x\alpha) z \beta$ for some $x,\ell,z\in\Sigma^*$ and $(\alpha,\beta)\in C$, i.e., whether $q\in(p,\alpha x \ell \theta(x\alpha) z \beta)$, and we check whether $\alpha x \ell \theta(x\alpha)\in H$ is a hairpin (possibly for some energy function $lin$, $con$, or $log$ and corresponding given parameters). If this is the case, add an $\varepsilon$-transition between $q_i$ and $q_j$, hence allowing for the word $w[1..i-1]w[j+1..|w|]$ to be read. Do this for all state-pairs and consider only transitions from the original automaton $A$. This way, we obtain an NFA $B$ that reads $[w]_{\hpdbp_{[C]}}$. To obtain $[w]_{\hpdbo_{[C]}}$, we create an unmodified copy of $A$ and move from $B$ to $A$ once we apply bounded hairpin-deletion. That way, bounded hairpin-deletion can only be applied once.

    To show containment of Problem \ref{problem:verification} for contextual lariat deletion in $P$, we give a rough polynomial upper bound on the time complexity. First of all, to construct the resulting automata with the above described procedure, we need to check whether the factor $w[i..j]$ is a hairpin regarding the considered energy function and parameters for each $i,j\in[|w|]$ with $i < j$. Assume the time to check hairpin-property is $T_H$. Then, this procedure takes $\bigo(w^2T_H)$. 
    
    Notice that $T_H$ must have a polynomial bound regarding the length of the checked factor for all three energy models considered here. First, for each splicing context $(\alpha,\beta)\in C$, we check whether the considered factor $w[i..j]$ has $\alpha$ as its prefix and $\beta$ as its suffix. Then, for each possible gap-length $|z|$ below the given gap-length bound $\delta$, we check whether the prefix $w[i..j-|z|-|\beta|]$ forms a hairpin in general. We can do that by computing the length of the longest possible stem and the remaining loop of $w[i..j-|z|-|\beta|]$ in linear time in $|w[i..j-|z|-|\beta|]|$. Finally, we compute whether the difference of these lengths results in a stable hairpin under the considered energy model. In total, this procedure takes polynomial time in $|w[i..j]|$.
    
    For the adaptations needed to construct the automaton for $[w]_{\hpdbo}$, we notice that the copy of $A$ takes, depending on implementation, $|A| = \bigo(|w|)$ time. Otherwise, the complexity is the same, resulting in $\bigo(|w| + |w|^2T_H) = \bigo(|w|^2T_H)$ time. 
    
    Now, after constructing the automaton, we just have to check in time linear in $|R|$ whether each word can be accepted. If that is the case, we return yes, otherwise return no. On total we obtain the time complexity $\bigo(|w|^2\cdot T_H + |R|)$.
\end{proof}

This concludes that templates can be verified efficiently with respect to contextual lariat deletion. 
Now, we show that template construction is rather computationally hard using a very similar construction to the one given for Theorem~\ref{thm:main-prob-decision-hard}.

\vspace{2mm}
\begin{proposition}\label{prop:parallelconstructability}
    Template constructability with respect to the parallel contextual lariat deletion under any of the three energy models is NP-complete with $|\Sigma| \geq 4$.
\end{proposition}
\begin{proof}
    NP containment of the problem follows similarly to the use of bracketed contextual deletion. 
    Given some solution $w$, checking whether $R \subseteq[w]_{\hpdb}$ can be done in $P$ by Proposition \ref{prop:hpa-verification-in-p}. 
    Hence, verification of a solution can be decided in $P$.

    One can reduce in polynomial time the NP-hard problem SCSe over some alphabet $|\Sigma'|$ to the Template Constructability Problem for parallel contextual lariat deletion over some alphabet $|\Sigma|$ of size $|\Sigma'|+2$.
    Due to its similarity to the proof provided in Theorem~\ref{thm:main-prob-decision-hard}, we focus on the differences and refer to the first proof for parts that follow from the same logic. 
    Let $(W,k)$ be an instance of SCSe such that $W\subset\Sigma^*$ is a finite set of words over some alphabet $\Sigma$ with $|\Sigma|\geq 2$. 
    We encode that instance to an instance of the template constructability, which consists of a finite language $R'$, a length bound $k'$, and parameters for an energy model to be considered. 
    Let $\#_s$ and $\#_e$ be letters not in $\Sigma'$, and let us define an antimorphic involution $\theta$ as  $\theta(\#_s) = \#_e$, $\theta(\#_e) = \#_s$, and $\theta(a) = a$ for all $a\in\Sigma'$. 
    The context is, again, $C = \{(\#_s,\#_e)\}$ and the gap bound is set to 0, forcing the right context to follow immediately after the hairpin. 
    We encode the words in $W$ into the words in $R'$ almost exactly in the same manner as done in Theorem~\ref{thm:main-prob-decision-hard}. 
    The only difference is that for each written $\#_e$, we now write $\#_e^2$ instead. 
    This is due to the need for the stem ending before the right context. 
    
    Model dependent energy parameters are set as follows: 
    for constantly-bounded loop length model, the loop length is bounded by $c = 1$; while for the other two models, we set $d_1 = d_2 = 1$ and $d = 0$. 
    Notice that, due to the definition of $C$, any deleted factor must start with $\#_s$ and end with $\#_e$. So a factor $\#_s\ta\#_e\#_e$ may always be deleted for some $\ta\in\Sigma'$, with $\#_s\ta\#_e$ forming the hairpin. For each longer factor $\#_s\ta_1\#_e\#_e...\#_s\ta_2\#_e\#_e$, for some $\ta_1,\ta_2\in\Sigma'$, also notice that the choice of the right context must always be the occurrence of $\#_e$ immediately after another $\#_e$, as otherwise no stem can be formed. So, a longer factor $\#_s\ta_1\#_e\#_e...\#_s\ta_2\#_e\#_e$ can only be deleted as a whole. By that, the reduction follows, using this construction, analogously to the proof of Theorem~\ref{thm:main-prob-decision-hard}.
\end{proof}

In the next section, we shift our attention to the general computational power of the introduced models for co-transcriptional splicing. 
By examining the language closure properties associated with these operations, we aim to deepen our understanding of their theoretical capabilities and potential occurrences in natural systems. This exploration provides valuable insights into the broader implications of the model.

%% file: sec-4-new-comp-prop.tex
Having settled the decision-problem landscape, this section investigates what kinds of language can be generated or recognized in the proposed models of co-transcriptional splicing. 
As a brief summary, this section progresses from basic closure results to non-closure and even undecidability for rather simple language classes, showcasing the expressive power of the contextual deletion operations proposed under all considered energy models. We start with the most permissive model (4.1). The class of regular languages is shown to be closed under iterated, even greedy, unbounded deletion, but already starting from a linear language one can obtain an undecidable language.
Section 4.2 then turns to the three bounded energy models. While the regular class stays closed under one-step and parallel deletion, both linear and context-free language classes lose closure. Already in the linear energy setting we obtain undecidability via a reduction from the Post-Correspondence Problem. The separate subsections of 4.2 compare the linear, constant, and logarithmic loop bounds, demonstrate further closure and non-closure properties and end with open questions that mark the frontier for future work.

\subsection{Language Properties of Bracketed Contextual Deletion}

First, we show that the class of regular languages is closed under both iterated bracketed contextual deletion and its right-greedy variant.

\vspace{2mm}
\begin{proposition}\label{prop:regclosure}
For a regular language $L$ and a finite set of contexts $C$, both languages $[L]_{\to_{[C]}^*}$ and $[L]_{\to_{[C]]}^*}$ are regular.
\end{proposition}
\begin{proof}
    Let $L$ be given by the NFA $M=(Q,\Sigma, q_1, \delta, F)$. For each pair of states $p,q\in Q$, define the language of words taking $p$ to $q$ as $L_{pq}=\{\ w \mid q\in\hat{\delta}(p,w)\ \}$. For each context $(u,v)\in C$ and for each pair $p,q$ such that $L_{pq}\cap u\Sigma^* v\neq \emptyset$, add an $\varepsilon$-transition from $p$ to $q$. For each $p,q$ pair recompute $L_{pq}$ to reflect any changes introduced by the new transitions. Repeat this until there are no more state pairs that satisfy the condition and are not yet connected by $\varepsilon$-transitions. This procedure must terminate, as the number of $\varepsilon$-transitions we can introduce is finite and because there can only be a finite number of paths between two states $q$ and $p$ where a factor $uxv\in\Sigma^*$, $(u,v)\in C$, occurs.
        
    This construction can be adapted for the right greedy variant by changing $L_{pq}\cap u\Sigma^* v\neq \emptyset$ to $L_{pq}\cap u\Sigma^*v \cap \overline{\Sigma^*v\Sigma^+}\neq \emptyset$.
\end{proof}

A 1-reversal-bounded stack and iterated \textit{right-greedy} contextual deletion can collaboratively encode a Turing machine as demonstrated below. 

\vspace{2mm}
\begin{proposition}\label{prop:linclosure}
    There exist a linear language $L$ and a finite set of contexts $C$ such that $[L]_{\to_{[C]]}^*}$ is undecidable.
\end{proposition}
\begin{proof}
Let $M$ be a Turing Machine, $Q$ be its (finite) set of states, and $Id_M$ be a set of its instantaneous descriptions (ID) of the form $\alpha(q, i)\beta$ for a current state $q \in Q$, the content $i \in \{0, 1\}$ of the cell where the input head is, and two binary words $\alpha, \beta \in \{0, 1\}^*$ that represent the non-blank portion of the input tape to the left and to the right of the cell, respectively. Using a binary alphabet $\{0', 1'\}$, disjoint from $\{0, 1\}$, and a homomorphism $h_{\#}: \{0, 1\}^* \to (\{0', 1'\} \cup \{\#_0, \#_1\})^*$ defined as $h_{\#}(0) = \#_0 0'$ and $h_{\#}(1) = \#_1 1'$, let us define the following language: 
    \begin{align*}
        L =&\ \{\ u_0 \$ v_2^R h_{\#}(u_2) \$\ \cdots\ v_{2k}^R \$ u_{2k-1} h_{\#}(v_{2k-1}^R)\$\ \cdots\ \$ u_3 h_{\#}(v_3^R) \$ u_1 h_{\#}(v_1^R) \mid \\
           &\ \ k \ge 0, u_0, u_1, \ldots, u_{2k-1}, v_1, \ldots, v_{2k} \in Id_M, \text{$u_0$ is an initial ID}, \\
           &\ \ \text{$u_j \vdash_M v_{j+1}$ for all $0 \le j < 2k$}, \text{and $M$ halts in $v_{2k}$}\}. 
    \end{align*}   
    This language is linear as it can be accepted by a one-turn deterministic PDA (see Figure~\ref{fig:lin-to-undec-bracketed} for an illustration). Indeed, the PDA pushes the input as it is up to the factor $\$v_{2k}^R\$$, which is distinguished from other factors sandwiched by $\$$'s by being free from $\#_0$ and $\#_1$, and then matching it with the remaining input in order to make sure that $v_{j+1}$ is a successor of $u_j$ according to the transition function of $M$ for all $0 \le j < k$. It also verifies that $u_0$ is an initial ID and $M$ halts in $v_{2k}$; these checks require no stack. 
    
    An element 
    \[
        w = u_0 \$ v_2^R h_{\#}(u_2) \$ \cdots \$ v_{2k}^R \$ \cdots \$ u_1 h_{\#}(v_1^R)
    \]
    of $L$ represents a valid halting computation $u_0 \vdash_M u_1 \vdash_M \cdots \vdash_M u_{2k-1} \vdash_M v_{2k}$ by $M$ only if $v_j = u_j$ for all $1 \le j < 2k$. $L$ contains also invalid computations in this sense. 
    Greedy deletion based on the contexts $(0\#_0, 0')$ and $(1\#_1, 1')$ iteratively filters out all these invalid ones. 
    See the factor $v_2^R h_{\#}(u_2)$. Let $u_2 = a_1 a_2 \cdots a_m$ and $v_2 = b_1 b_2 \cdots b_n$ for some $a_1, \ldots, a_m, b_1, \ldots b_n \in \{0, 1\}$. Then this factor is $b_n \cdots b_2 \underline{b_1 \#_{a_1}} \ \underline{a_1'} \#_{a_2} a_2' \cdots \#_{a_m} a_m'$, where the underlines indicate the sole factors to which the greedy bracketed contextual deletion can be applied, but it can be actually applied if and only if $b_1 = a_1$, resulting in $b_n \cdots b_3 \underline{b_2 \#_{a_2}} \ \underline{a_2'} \#_{a_3} a_3' \cdots \#_{a_m} a_m'$ (see Figure~\ref{fig:lin-to-undec-bracketed-equality-condition} for an illustration of this property). 
    A word in $u_0 \$^*$ is thus obtainable from $w$ if and only if $u_j = v_j$ for all $1 \le j < 2k$. 
    This means that $[L]_{\to_{[C]]}^*} \cap \Sigma^* \$^*$ is undecidable. 
    Its intersection with a regular language being undecidable, $[L]_{\to_{[C]]}^*}$ must be also undecidable. 
\end{proof}

\begin{figure}[h]
    \centering
    \includegraphics[scale=0.4]{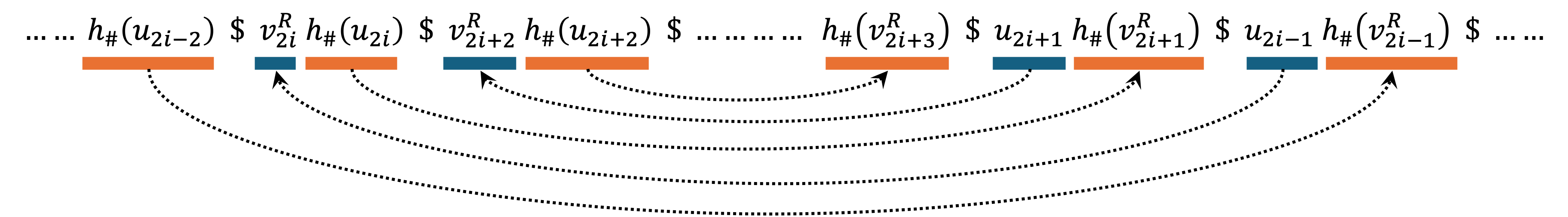}
    \caption{An illustration of the general construction used in the proof of Proposition~\ref{prop:linclosure}. Same colored blocks conncted by an arrow indicate two configurations $u_i$ and $v_{i+1}$ for which $u_i\vdash_M v_{i+1}$. Notice that we can guess some configurations on the left side and generate the corresponding right side using a 1-turn PDA.}
    \label{fig:lin-to-undec-bracketed}
\end{figure}

\begin{figure}[h]
    \centering
    \includegraphics[scale=0.60]{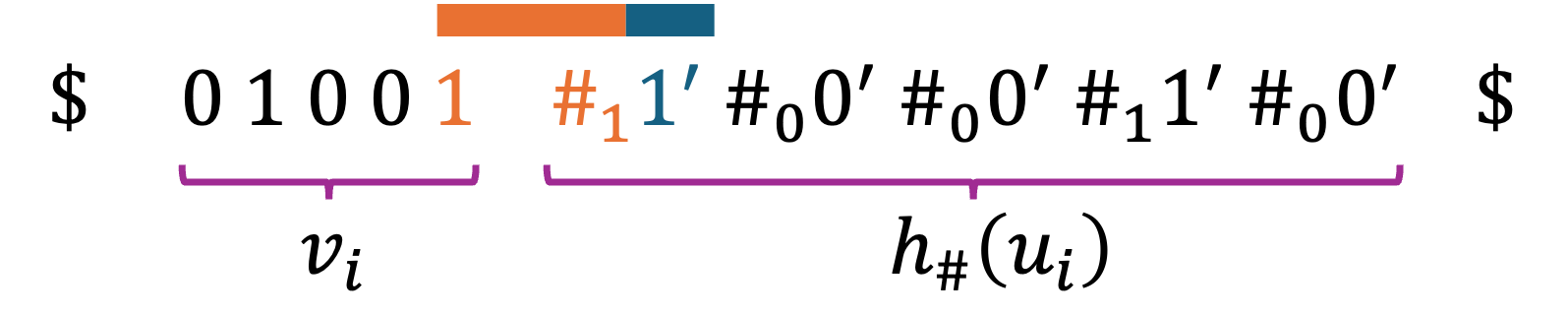}
    \caption{An illustration of the fact that, in Proposition~\ref{prop:linclosure}, we can only remove everything between $\$$ symbols using bracketed contextual deletion if and only if $v_i = u_i$. The left context is marked by the orange bar while the right context is marked with the blue bar. Notice that left contexts can only form if $v_i$ and $h_{\#}(u_i)$ are both not empty. By the definition of $h_{\#}$, if a left context is found, then by the definition of right-greediness, directly after that the corresponding right context will occur in the same block surrounded by $\$$ signs.}
    \label{fig:lin-to-undec-bracketed-equality-condition}
\end{figure}

Right-greediness is critical in this proof for filtering; it remains open how computationally powerful $[L]_{\to_{[C]}^*}$ (not right-greedy) can be for a linear language $L$. 

To further investigate the tight bounds on computational capabilities of the right-greedy variant, let us consider the class of \textit{1-counter} languages, that is, those recognized by a pushdown automaton whose stack is implemented on a unary alphabet $\{a\}$ (the bottom is indicated by a special letter $\perp$). 
This class turns out not to be closed, either. 
More strongly, we will show that an undecidable language can be obtained from some 1-counter language and context-set $C$. 

\vspace{2mm}
\begin{proposition}\label{prop:onecounternotclosed}
    The class of 1-counter languages $L_{1C}$ is not closed under iterated right-greedy bracketed contextual deletion.
\end{proposition}
\begin{proof}
    Consider the language
    $$ L = \{\ \ta^{k_0}\#\tb^{2k_0}\ (\$\ta)^{k_1}\#\tb^{2k_1}\ (\$\ta)^{k_2}\#\tb^{2k_2}\ ...\ (\$\ta)^{k_n}\#\tb^{2k_n}\ |\ n, k_0,\ ...\ , k_n\in\N\ \} $$
    over the alphabet $\Sigma = \{\ta,\tb,\#,\$\}$. This language can be recognized using a PDA with a unary stack. Let $C = \{(\tb\$,\ta)\}$. Then, for each factor $\#\tb^{2k_i}(\$\ta)^{k_{i+1}}\#$ we have 
    \[
        \#\tb^{2k_i}(\$\ta)^{k_{i+1}}\# \to_{[C]]}^* \#\#
    \]
    if and only if $2k_i = k_{i+1}$. 
    Hence, taking the intersection of $[L]_{\to_{[C]]}^*}$ with the regular language $\ta\#^*\tb^*$ results in the non-context-free language $\{\ \ta\#^n\tb^{2^n}\ |\ n \ge 1\ \}$ (its set of Parikh vectors is not a semi-linear set). 
    Now it suffices to note that 1-counter languages are context-free and the class of 1-counter languages is closed under intersection with a regular language.
\end{proof}

\vspace{2mm}

\begin{proposition}
    There exists a 1-counter language $L$ and a finite set of contexts $C$ such that $[L]_{\to_{[C]]}^*}$ is undecidable.
\end{proposition}
\input{exists-1c-lang-unb-iterat-set-undec}

This concludes the results regarding language closures under iterated bracketed contextual deletion. 
Regarding the non-greedy cases, the question regarding its computational power remains open for 1-counter languages. 
In particular, it is still open whether the class of 1-counter languages is closed under iterated contextual splicing. 

\subsection{Language properties of contexual lariat deletion}

Let us move towards more realistic model of co-transcriptional splicing, that is, 1-step contextual lariat deletion and its parallel extension.
The results of this subsection focus on language closure properties under this operation under the proposed three energy models.

\subsubsection{Set of stable hairpins}

Unlike the bracketed contextual deletion, the contextual lariat deletion involves a stable hairpin. 
It is hence significant to know how computationally easy/hard it is to tell whether a given hairpin is stable or not. 

The stability test can be easily done in the constantly-bounded loop-length model by using a one-turn pushdown automaton, which characterizes the class of linear languages. 
Such an automaton operates in three phases: firstly it reads and pushes some prefix $x$ of an input onto the stack; then it non-deterministically starts reading the loop $\ell$ bounded with constant size; and finally, it starts popping $x$ from the stack to check whether the remaining suffix is $\theta(x)$ or not. 

\vspace{2mm}
\begin{lemma}
    For any constant $c\in\N$, $H_{{\rm con},(c)}$ is linear. 
\end{lemma}
\vspace{2mm}

The test is not as easy in the other models, intuitively because therein, one half of a stem must be not only matched with the other half but also compared in length with the loop in-between. 

\vspace{2mm}
\begin{lemma}\label{lemma:H-lin-not-context-free}
    There exist parameters $d_1, d_2, d$ such that $H_{{\rm lin}, (d_1, d_2, d)}$ is not context-free.
\end{lemma}
\begin{proof}
    Let $\Sigma = \{0,1,a,b\}$ and set $\theta$ such that $\theta(1) = 0$, $\theta(0) = 1$, $\theta(a) = b$, and $\theta(b) = a$. 
    Set $d_1 = d_2 = 1$ and $d = 0$; these parameters in the linear loop penalty model mean that in order for a hairpin to be stable, its stem is not allowed to be shorter than its loop. 
    Suppose $H_{\rm lin}$ were context-free; then so is $H' = H_{lin} \cap 0^+1^+a^+0^+1^+$.
    
    Let $p\in \N$ be the constant for $H'$ in the well-known Bar-Hillel pumping lemma for context-free languages~\cite{Kreowski1979}. Then, by the pumping lemma we get that every word in $H'$ has a decomposition $uxvyw$ such that $|xvy|\leq p$, $|xy|>0$ and $ux^ivy^iw\in H'$ for any $i\geq 0$. It is easy to see that $0^p1^pa^{2p}0^p1^p\in H'$, and we will argue that there is no decomposition for this word satisfying the lemma, contradicting the context-freeness of $H'$. We can immediately conclude that $xvy\notin (0+1)^+a^{2p}(0+1)^+$, due to the condition $|xvy|\leq p$. Three cases remain: (1) $xvy\in a^+$,  (2) $xvy\in (0+1)^+a^+$ or (3) $xvy\in (0+1)^+$ (and the analogous cases of (2) and (3), when $xvy$ is factor of the second half of the word). It is straightforward that in case (1) we have $xy\in a^+$. Setting $i=2$ we get $0^p1^pa^{2p+|xy|}0^p1^p\in H'$, but this is a contradiction, because $a^{2p+|xy|}$ is all part of the loop, which is longer than the longest possible stem $0^p1^p$; hence this hairpin would not be stable.
    In case (2) we set $i=2$ again and get a word in $(0+1)^{2p+m} a^{2p+n}0^p1^p$, where $m+n=|xy|$. Note that the maximal length of the stem is still $2p$, because any longer suffix cannot match completely with any prefix, which means that the loop length is at least $2p+|xy|$, a contradiction.
    In case (3), setting $i=0$ again we get that the loop length of at least $2p$ is more than the maximal stem length of $2p-|xy|$, a contradiction.
\end{proof}

\vspace{2mm}
\begin{lemma}
    There exist parameters $d_1, d_2, d$ such that $H_{{\rm log}, (d_1, d_2, d)}$ is not context-free.
\end{lemma}
\begin{proof}
    Let $\Sigma = \{0,1,a,b\}$ and set $\theta$ such that $\theta(1) = 0$, $\theta(0) = 1$, $\theta(a) = b$, and $\theta(b) = a$. Set $d_1 = d_2 = 1$ and $d = 0$. Suppose $H_{log}$ was context-free for these parameters. Similar to Lemma \ref{lemma:H-lin-not-context-free}, let $$H' = H_{log} \cup 0^+1^+a^+0^+1^+$$ be an intersection of $H_{log}$ with that specific regular language, again, restricting the form of all considered hairpins. By the assumption, $H'$ should be context-free as well. Let $p\in\N$ be the constant given by the pumping lemma for context-free languages \cite{Kreowski1979}. Let $$w = 0^p1^pa^{2^{2p}}0^p1^p.$$
    Now, the exact same arguments as in Lemma \ref{lemma:H-lin-not-context-free} can be applied to obtain that $H'$ is not context-free.
\end{proof}

The last two lemmas shall be used to show that the class of linear languages is not closed under 1-step or parallel contextual lariat deletion under linear or logarithmic loop penalty model. 
They also raise the problem of which language classes the set of stable hairpins lies in, though it is not directly related to what we will discuss in the succeeding sections. 
Let us settle this problem first. 
The notion of \emph{nondeterministic pushdown automata augmented with one counter (NPCM)} plays an important role. 
They work analogously to nondeterministic pushdown automata with the addition of a counter (unary stack) that may be handled in each transition as well. 
A \emph{1-reversal-bounded NPCM} (denoted \emph{NPCM(1)}) is a NPCM whose counters are assumed to be $1$-reversal-bounded, i.e., once they start decrementing, they cannot increment anymore. 
The following result shows that a NPCM(1) suffices to obtain exactly the language $H_{lin}$ for arbitrary parameters.
This will help us due to the fact that - as opposed to NPCM with no counter restrictions - Ibarra has shown in~\cite{Ibarra78} that a language accepted by a 1-reversal-bounded NPCM is semilinear and has efficient algorithms for constructing intersections with regular languages and checking emptiness.

\vspace{2mm}
\begin{lemma}\label{lemma:hlinNPCM}
    For any parameters $d_1, d_2, d$, the language $H_{{\rm lin}, (d_1, d_2, d)}$ can be accepted by a $1$-reversal-bounded NPCM; hence, it is semilinear.
\end{lemma}
\begin{proof}
    Let $\Sigma$ be some alphabet and $\theta$ be an antimorphic involution on $\Sigma$. Let $d_1,d_2\in\N$ and $d\in\N_0$ be some parameters for $H_{lin}$. We consider the language $H_{lin}$. We construct a 1-reversal-bounded NPCM $A$ such that $L(A) = H_{lin}$. $A$ is split up in $3+c+d_1+d_2$ states $Q = \{q_1,q_2,q_3,q_{con,1},...,q_{con,c},q_{fac1,1},...,q_{fac1,d_1-1},q_{fac2,1},...,q_{fac2,d_2-1}\}$. 
    In $q_1$, the stem of the hairpin is being written. For each letter $a\in\Sigma$ that is read in $q_1$, $\theta(a)$ is added to the stack and the counter is increased by $d_2$ using $\varepsilon$-transitions between $q_1$ and $q_{fac2,1}$, $q_{fac2,i}$ and $q_{fac2,i+1}$ for each $i\in[d_2-2]$, and $q_{fac2,d_2-1}$ and $q_1$ where the value of the counter is increased by $1$ per transition. We can non-deterministically choose to start writing the loop of the hairpin. We add an $\varepsilon$-transition from $q_1$ to $q_{con,1}$. Now, the constant $c$ is being processed. For each $i\in[c-1]$, we add an $\varepsilon$-transition from $q_{con,i}$ to $q_{con,i+1}$ that reduces the value of the counter by $1$. If at any point this is not possible, the word is not in the language. Finally, an $\varepsilon$-transition between $q_{con,c}$ and $q_2$ finished the handling of the constant and allows for writing the loop. For each letter $a\in\Sigma$ we read in $q_2$, we reduce the counter by $d_1$. Analogously, to increase the counter in $q_1$ by $d_2$, we use the states $q_{fac1,1}$ to $q_{fac2,d_1-1}$ to handle this. If at any point, this would result in the counter going negative, i.e., we cannot reduce the counter anymore, the word is rejected and is by that not in the language. Finally, we can start writing the right part of the stem at any time by moving from $q_2$ to $q_3$ with an $\varepsilon$-transition. In $q_3$, we can only read letters from the stack if they are in the word. Once the stack is empty, we cannot read any more letters and accept the word read so far. By this construction, we obtain that there exists a NPCM(1) that accepts $H_{lin}$.
\end{proof}

\subsubsection{In the linear loop penalty model}

As announced above, let us now use Lemma~\ref{lemma:H-lin-not-context-free} to show that the class of linear languages is not closed under 1-step and parallel contextual lariat deletion in the linear loop penalty model. 
Let 
\[ 
L = \{\ 0^i1^ka^n0^s1^t\$1^t0^sa^n1^k0^i\$ \mid i,k,n,s,t\in\N_0\}.
\]
This language is linear. 
Let us apply the 1-step lariat deletion according to the antimorphic involution $\theta$ that swaps 0 with 1 and $a$ with $b$ under the linear loop penalty model that prevents a hairpin from having a longer loop than a stem (the exact setting in the proof of the lemma). 
Consider the single context $C = \{(\$,\$)\}$ and set the gap bound to 0. 
Deleting one stable hairpin from words in $L$ and filtering the words thus obtained by the regular ``formatting'' language $1^*0^*a^*1^*0^*$ yields $H'$, the set of all stable hairpins in this format, but $H'$ has been shown in the lemma not to be context-free. 
The same argument works for $\hpdblinp$ as any word in $L$ involves only one position where the lariat deletion can be applied. 

\vspace{2mm}
\begin{proposition}\label{prop:linear-haipin-deletion-liner-languages-not-closed}
There exist a linear language $L$, a finite set of contexts $C$, and parameters $d_1, d_2, d$ such that neither $[L]_{\hpdblino}$ nor $[L]_{\hpdblinp}$ is linear.
\end{proposition}
\vspace{2mm}

In contrast, we can show that the class of regular languages is effectively closed under both of these operations for arbitrarily set parameters. 
We need the following lemma.

\vspace{2mm}
\begin{lemma}\label{lemma:linear-hpd-regular-language-applicability-decidable}
    For a regular language $L$ and parameters $d_1, d_2, d$, we can decide whether there exists $w\in L$ such that $w\ \hpdblin_{(d_1, d_2, d)}\ \varepsilon$.
\end{lemma}
\begin{proof}
    By Lemma~\ref{lemma:hlinNPCM}, we know that $H_{{\rm lin}, (d_1, d_2, d)}$ is in the class of languages accepted by a 1-reversal-bounded NPCM. 
    This class is closed under intersection with regular languages~\cite{DALEY200319} and emptiness is decidable in it~\cite{Ibarra78}. 
    Hence, one can effectively check $H_{{\rm lin}, (d_1, d_2, d)} \cap L\neq \emptyset$, which is equivalent to the question in the statement.
\end{proof}

Let us now prove the closure property. 

\vspace{2mm}
\begin{proposition}\label{lemma:linear-hairpin-deletion-regular-languages-closed}
    For a regular language $L$, a finite set of contexts $C$, and parameters $d_1, d_2, d$, both languages $[L]_{\hpdblino}$ and $[L]_{\hpdblinp}$ are regular,  and finite automata accepting them can be constructed from a finite automaton for $L$.
\end{proposition}
\begin{proof}
    Let $\theta$ be an antimorphic involution on $\Sigma$ and $g$ be a gap bound. 
    Similar to the proof of Proposition \ref{prop:regclosure}, let $L$ be given by some NFA $A = (Q,\Sigma,q_1,\delta,F)$. Assume w.l.o.g. that $Q = \{q_1, q_2, ... , q_m\}$. For each pair of states $p,q\in Q$, define the language of words from the state $p$ to the state $q$ as $L_{pq} = \{\ w\mid q\in\delta(p,w)\ \}$. Construct a new automaton $B$ which consists of $A$ and a copy of $A$, called $A' = (Q',\Sigma,q_1',\delta',F)$, for which we have that $Q' = \{q_1', q_2', ... , q_m'\}$ and each transition in $\delta'$ connects the elements of $Q'$ as $\delta$ connects the elements in $Q$. Then, if for some $w\in L_{pq}$ we have that $w \ \hpdblin\ \varepsilon$ (decidability shown in Lemma \ref{lemma:linear-hpd-regular-language-applicability-decidable}), then add an $\varepsilon$-transition between $p$ and $q'$ from the original automaton $A$ to the copy $A'$. Then, if between two states a hairpin can be removed according to the set parameters, then we can move to the copy of the second state. From there on, we have the exact behavior as in $A$, just without any other $\varepsilon$-transition that represents a contextual lariat deletion step. As this can be done for all state-pairs, we can choose arbitrarily when to do a valid contextual lariat deletion step. Thus, the language of $B$ is exactly the 1-step contextual lariat deletion set $[L]_{\hpdblino}$. So, $[L]_{\hpdblino}$ is a regular language. 
    The same can be done analogously for all other energy models.
    To adapt this proof for parallel contextual lariat deletion, we can just omit the construction of the copied automaton $A'$ and stay in the adapted automaton $B$. That way, any parallel deletion may be applied while reading the word.
\end{proof}

\begin{remark}
In contrast to the result for iterated bracketed contextual deletion, we do not use the newly created $\epsilon$-transitions to obtain any more transitions between states. That way, we can represent parallel contextual lariat deletion without accidentally considering an iterative variant of the operation. If, however, we considered iterated contextual lariat deletion, we may obtain the same closure result as before by iteratively allowing for newly created $\epsilon$-transitions to be used to obtain even more.\end{remark}\vspace{2mm}

A final question that can be asked is whether we can obtain any result that involves undecidability using bounded contextual lariat deletion over the linear energy model, preferably 1-step or parallel bounded contextual lariat deletion. Actually, starting from some language $L$ such that $L = L(A)$ for some NPCM(1) $A$, it can be shown that the applicability of contextual lariat deletion under the linear energy model and the question whether $\varepsilon \in [L]_{\hpdblino}$ or $[L]_{\hpdblinp}$ is generally undecidable. Consider the following reduction from the Post Correspondence Problem.

\vspace{2mm}
\begin{theorem}
    Given some NPCM(1) $A$, there exist parameters for linear  contextual lariat deletion such that it is undecidable to answer whether 
    \begin{itemize}
        \item $\varepsilon\in [L]_{\hpdblino}$ or $\varepsilon\in [L]_{\hpdblinp}$ as well as
        \item $[L]_{\hpdblino} = L(A)$ or $[L]_{\hpdblinp} = L(A)$, i.e., whether no contextual lariat deletion under the linear energy model can be applied.
    \end{itemize}
\end{theorem}
\begin{proof}
    We reduce the Post Correspondence Problem (PCP) to the above mentioned problems. Let $\Sigma = \{0,1,\alpha,\beta,a,b\}$ and assume $\{\ (x_1,y_1),(x_2,y_2),...,(x_n,y_n) \mid x_i,y_i\in\{0,1\}^*, i\in[n]\ \}$ to be some PCP instance. We define the linear grammar $G = (V,\Sigma,S,P)$ over $\Sigma$ with non-terminals $V$, a start-symbol $S$ and productions $P$ by the following: Set $V = \{S,T,A\}$ and define the productions $S \rightarrow \alpha T \alpha \beta$, $T \rightarrow x_i T y_i^R$ for each $i\in[n]$, $T \rightarrow A$, $A \rightarrow aA$, and $A \rightarrow a$. Then, with $G$, all words that can be constructed have the form $$\ \alpha x_{i_1}x_{i_2}...x_{i_m}\ a^k\ y_{i_m}^R...y_{i_2}^Ry_{i_1}^R\ \alpha\beta$$ for any $k\in\N_0$ and any sequence $(i_1,i_2,...,i_m)\in[n]^m$ for $m\in\N$. Using the additional 1-reversal-bounded counter, we can check whether $|x_{i_1}x_{i_2}...x_{i_m}| = k-1$ and restrict all words produced by $G$ to exactly those. A translation to some NPCM(1) $A$ that has that language can be constructed immediately. Assume $\theta$ to be an antimorphic involution with $\theta(0) = 0$, $\theta(1) = 1$, $\theta(\alpha) = \alpha$, $\theta(\beta) = \beta$, $\theta(a) = b$, and $\theta(b) = a$. Additionally, assume the constant $d = 0$ and the factors $d_1 = d_2 = 1$. Finally assume the gap-length bound $\delta$ of size $0$. Assume the context-set $C = \{(\alpha,\beta)\}$. Any word in $L(A)$ has the form as described above. In particular, no word has the letter $b$ in it. Due to the context-set $C$, we know that linear contextual lariat deletion can only be applied to the whole word at once. Additionally, due to the definition of $\theta$, we know that the letter $a$ cannot occur in any part of the stem. Due to the fact that $|\alpha x_{i_1}...x_{i_m}| = k$, we know that linear contextual lariat deletion can only be applied if $|\alpha x_{i_1}...x_{i_m}| = |y_{i_m}...y_{i_1} \alpha|$ and thus only if $\alpha x_{i_1}...x_{i_m} = \theta(y_{i_m}...y_{i_1} \alpha)$. But then, we know that the sequence of indices $(i_1,...,i_m)$ is a valid solution for the PCP instance. Hence, we can only have $\varepsilon\in [L]_{\hpdblino}$ or $[L]_{\hpdblino} \neq L(A)$ (analogously $\varepsilon\in [L]_{\hpdblinp}$ or $[L]_{\hpdblinp} \neq L(A)$) if and only if the PCP instance has a solution. As this problem is undecidable, we know that answering these two questions must also be undecidable. This concludes this proof.
\end{proof}

This concludes the current results for contextual lariat deletion over a linear energy model. We showed that the class of regular languages is closed under 1-step bounded linear contextual lariat deletion, those of linear languages as well as context-free languages are not, that hairpins under a linear energy model can be modeled using non-deterministic pushdown automata augmented with a 1-reversal bounded counter, and that we can obtain undecidability results from those automata with bounded linear contextual lariat deletion. It is still open whether undecidability can be obtained from context-free or linear languages using contextual lariat deletion with a linear energy model.

\subsubsection{In the constantly-bounded loop length model}

The language of constant-bounded hairpins is linear and therefore can be accepted by push-down automata, the arguments in Lemma \ref{lemma:linear-hpd-regular-language-applicability-decidable} and by that also the ones from Proposition \ref{lemma:linear-hairpin-deletion-regular-languages-closed} can be reused to obtain the next result.

\vspace{2mm}
\begin{corollary}\label{corollary:regularclosed}
    Let $L$ be some regular language. Then, for any parameters, we have that $[L]_{\hpdbcono}$ and $[L]_{\hpdbconp}$ are regular languages.
\end{corollary}
\vspace{2mm}

Hence, we can focus our investigation on more expressive language classes. Most interestingly, it can be shown that context-free class is not only not closed under constantly-bounded contextual lariat deletion, but that we can also obtain an undecidable language by a context-free language.

\vspace{3mm}
\begin{theorem}\label{theorem:constant-bound-cf-to-undec}
    There exists a context-free language $L$, a context-set $C$, and parameters for constantly-bounded contextual lariat deletion, such that $[L]_{\hpdbconp}$ is an undecidable language.
\end{theorem}
\begin{proof}
    Let $\Sigma = \{0,1,\#,\$,\$',\alpha,\beta\}$ be some alphabet and consider the context-set $C = \{(\alpha,\beta)\}$. Additionally, assume some constant upper bound $c\in\N$ and assume an upper bound of $0$ for the gap between the right context and the hairpin. Set the antimorphic involution $\theta$ such that $\theta(\ta) = \ta$ for all $\ta\in\Sigma\setminus\{\$,\$'\}$ and assume $\theta(\$) = \$'$ and vice versa. Set the constant loop bound to $k\in\N$.
    For an arbitrary TM $M$, construct the language $L$ consisting of words $w_0\#\alpha w_1\#\cdots w_n\$^k u_m\#\cdots u_2\#u_1\alpha\beta$, where each $w_i\in\{0,1\}^*$ and $u_j\in\{0,1\}^*$ for $i\in[n]$ and $j\in[m]$ is a binary configuration of $M$ and there is a valid transition of $M$ from $w_{2i}$ to $w_{2i+1}^R$ for each $i$, and there is a valid transition from $u_{2i+1}$ and $u_{2i+2}^R$ for each $i$, for some $m,n\in\N$. Furthermore, we require that $w_n$ and $u_m$ are final configurations, i.e., they contains a final state of $M$. This language is context-free, since a PDA can check pairs of adjacent configurations for correct transitions and can verify the regular condition that $w_n$ and $u_m$ are final configurations. Now, we can apply constantly-bounded contextual lariat deletion to obtain the language $[L]_{\hpdbconp}$. Notice that there is only one context $(\alpha,\beta)\in C$. Additionally, notice that there is only one occurrence of $\beta$. Additionally, notice that there is only pair of occurrences of the left context $\alpha$ that could form a hairpin. Hence, if $[L]_{\hpdbconp}$ contains a word with a removed hairpin, then it will be of the form $w_0\#$, for some $w_0\in\{0,1\}^*$. Also, we see that there are no occurrences of $\$'$ in any word in $L$. This prevents the factor $\$^k$ to be part of any stem. Hence, the factor $\$^k$ always forms a loop that reaches the loop length bound $k$. So, we know that $\alpha w_1 \# ... \# w_n = \theta(u_m \# ... \# u_1 \alpha)$ must form the corresponding stem of the hairpin. By the definition of $\theta$, this is only the case if $\alpha w_1 \# ... \# w_n = \alpha u_1 \# ... \# u_m$ and $m=n$. By the definition of $L$, we know that there is a valid transition between the configurations $w_{2i}$ and $w_{2i+1}^R$ as well as $u_{2i+1}$ and $u_{2i+2}^R$. As $w_j = u_j$, for all $j\in[n]$, if $\alpha w_1 \# ... \# w_n = \alpha u_1 \# ... \# u_m$, then there must also exist a valid transition between the configurations $w_{2i+1}^R$ and $w_{2i+2}$ as well as $u_{2i}^R$ and $u_{2i+1}$ (see Figure~\ref{fig:constant-bounded-cf-to-undec} for an illustration). Hence, using that argument inductively, there is a sequence of valid transitions between all configurations $w_0$ up to $w_n$ as well as $u_1$ up to $u_m$ (iterating between reversed and non-reversed encodings). Let $L_{reg}$ be the language represented by the expression $(0|1)^*\#$. By the previous observations, we know that the intersection $[L]_{\hpdbconp}\cap L_{reg}$ contains only words $w_0\#$, $w_0\in\{0,1\}^*$, where $w_0$ is the initial configuration of some valid computation of $M$. Hence, $[L]_{\hpdbconp}\cap L_{reg}$ is undecidable. By that, $[w]_{\hpdbconp}$ must be undecidable as well.
\end{proof}

\begin{figure}[h]
    \centering
    \includegraphics[scale=0.6]{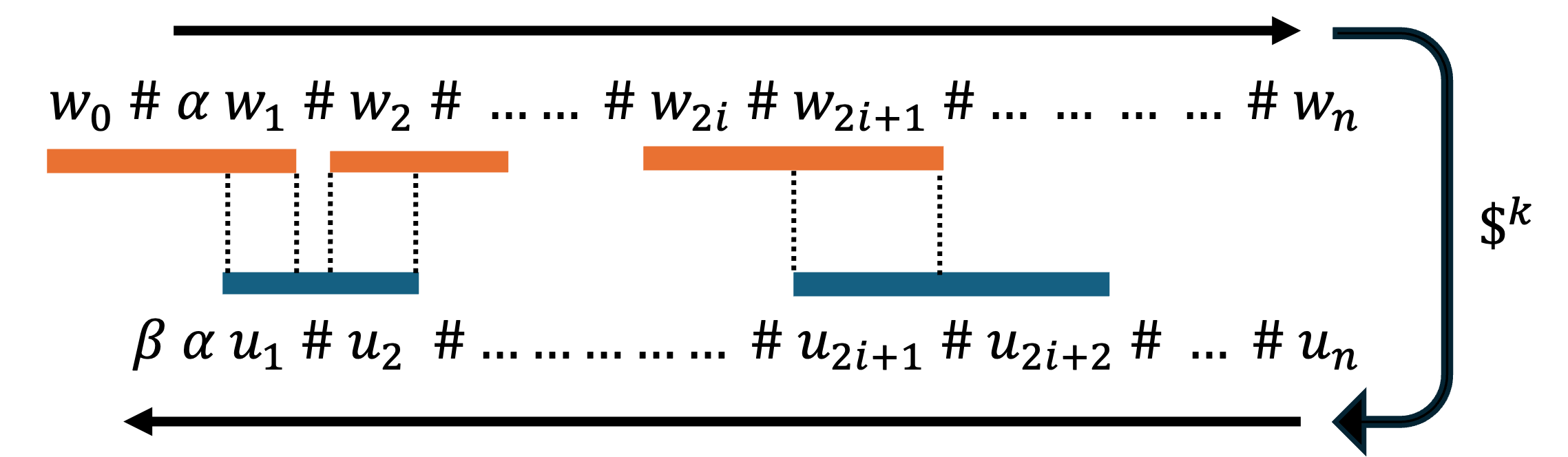}
    \caption{An illustration of the inductive argument in the proof of Theorem~\ref{theorem:constant-bound-cf-to-undec}. The arrows mark the removed hairpin. The dotted arrows mark equal configurations, in the case of a hairpin being formed. The colored bars mark configurations connected by valid transitions.}
    \label{fig:constant-bounded-cf-to-undec}
\end{figure}

In contrast to the linear energy model, this leaves us only with the question of what happens if we consider an intermediate language class such as linear languages. Whether that languages class is closed under constantly-bounded contextual lariat deletion or whether we can even obtain an undecidable language is still open. Clearly, we can construct linear languages to obtain the class of constant-bounded hairpins for some given parameters. But as the intersection of linear languages is not generally closed under linear languages, a closure under the operation of constantly-bounded contextual lariat deletion does not trivially follow. This concludes the results regarding this energy model.

\subsubsection{In the logarithmic loop penalty model}

Continuing with the logarithmic model, as in the case of the linear energy model, the following property can be shown.

\vspace{2mm}
\begin{lemma}\label{lemma:log-hpd-regular-language-applicability-decidable}
    For a regular language $L$, a finite set of contexts $C$, and parameters $d_1, d_2, d$, we can decide whether there exists $w\in L$ such that $w\ \hpdbloge_{[C], (d_1, d_2, d)}\ \varepsilon$.
\end{lemma}
\begin{proof}
 Recall the procedure from Algorithm~\ref{algo:inverse} for constructing an NFA that accepts words obtainable from a given starting word through contextual lariat deletion. The input to that algorithm was a word, but we immediately constructed a DFA from it and looked for pairs of states $p,q$ such that some hairpin (in that case, simply a word that starts with a left context and ends in the corresponding right context) is accepted starting from $p$ finishing in $q$. Between each of those pairs, we added an $\varepsilon$-transition to accept words obtained by removing that hairpin. Here we use the same basic idea, but we have to account for the fact that hairpins have a more complicated structure in the logarithmic model, so we need to adapt the step of deciding for a pair of states $p,q$ whether there is any hairpin accepted between them. This means that we need to be able to test $H\cap L'\neq \emptyset$ for the set of hairpins $H$ and a regular language $L'$. Under the logarithmic energy model, each hairpin has a loop that contributes $\Theta (\log n)$ to the energy, where $n$ is the loop length. If there exists some hairpin $x\ell \theta(x)\in H\cap L'$ with $|\ell|\geq |Q|$, where $|Q|$ is the number of states in the minimal DFA accepting $L'$, then by a usual pumping argument there exists some $\ell'$ with $|\ell'|<|Q|$ such that $x\ell'\theta(x)\in L'$ and because the loop in the latter is shorter, we also have $x\ell'\theta(x)\in H$. Similarly, if $|x|\geq N+|Q|$, for some $N$, then we can write $x=x_1x_2$ with $|x_1|=N$ and we can again use a pumping argument to obtain that there exist some $x_2', x_2''$ with $|x_2'|<|Q|$ and $|x_2''|<|Q|$ such that $x_1x_2'\ell'x_2''\theta (x_1)\in L'$. Setting $N=\log(3|Q|)$, we get that if $H\cap L'$ is not empty then it must contain a word of length at most $\log(3|Q|)+3(|Q|-1)$. For any given word $w$ it is easy to check whether it is in $L'$ and whether it is a valid hairpin for the given parameters. Testing $w\in H\cap L'$ for all words $w$ up to the aforementioned upper bound is therefore effective, which means that the lemma holds.
\end{proof}

From here, again following the argument from Proposition~\ref{lemma:linear-hairpin-deletion-regular-languages-closed} we get the following.

\vspace{2mm}
\begin{corollary}\label{corollary:reg}
    For a regular language $L$ and any parameters, both languages $[L]_{\hpdblogo}$ and $[L]_{\hpdblogp}$ are regular.
\end{corollary}
\vspace{2mm}

In exactly the same way as for Proposition~\ref{prop:linear-haipin-deletion-liner-languages-not-closed}, we obtain the following result.

\vspace{2mm}
\begin{corollary}\label{corollary:notclosedcfl}
    There exist a linear language $L$ and parameters such that neither $[L]_{\hpdblogo}$ nor $[L]_{\hpdblogp}$ is context-free.
\end{corollary}
\vspace{2mm}

In contrast to the linear energy model, it is still open whether we can obtain some undecidability result regarding the applicability of logarithmic-bounded hairpins. Also, it is still open whether we can obtain an undecidable language by applying logarithmic contextual lariat deletion to a linear- or context-free language. 

\vspace{2mm}

This concludes all current results regarding contextual lariat deletion under different energy-model assumptions as well as concludes the current selection of results regarding the computational power of all models of co-transcriptional splicing introduced in this paper.

%% file: exists-1c-lang-unb-iterat-set-undec.tex
\begin{proof}

    We show this by taking some right-bounded Turing machine $M$, constructing a specific 1-counter language $L$, and defining a context-set $C$ based on $M$ such that $[L]_{\to_{[C]]}^*}$ intersected with a specific regular language contains only words if $L(M)$ is not empty.

    Let $M =(Q,\Sigma,\Gamma,\delta,q_1,\text{\textvisiblespace},F)$ where without loss of generality $Q := \{q_1,...,q_n\}$ is the set of states for some $n\in\N$, $\Sigma$ is the input alphabet, $\Gamma \supset \Sigma$ is the tape alphabet, $\delta : (Q\setminus F)\times\Gamma \rightarrow Q\times\Gamma\times\{\mathtt{N,L,R}\}$ is the transition function, $q_1\in Q$ is the initial state, $\text{\textvisiblespace}\in(\Sigma\setminus\Gamma)$ is the blank symbol, and $F\subset Q$ is the set of final/halting states.

    For the following construction, we need some general definitions. First, we define $\binary(i)\in\{0,1\}^*$ to be the binary representation of $i\in\N$. For some binary word $w\in\{0,1\}^*$ we define $\decimal(w)$ to be the decimal representation of $w$ with the most significant bit first. For all $i,j\in\N$ with $j > |\binary(i)|$ we let $\binary_j(i)$ be the binary representation of $i$ padded on the left with $j-|bin(i)|$ many $0$'s. 
    We set the constant $c = |\binary(n)|$ based on $M$ and define the binary encoding of states by $\enc(q_i) = \binary_c(i)$ for all $q_i\in Q$.

    As $M$ is right-bounded, we can represent each configuration of $M$ with a single binary number, where the information of the head's location and the current state is encoded in the binary string. Due to technical reasons, we assume the first (most significant) bit of the tape contents to be $1$ at all times. As $M$ is not left bounded, we can always make this assumption without losing expressiveness. Let $L_{\mathtt{STATES}} := \{\ \enc(q_i)\ |\ q_i\in Q\ \} \cup \{0^c\}$ where $0^c$ represents that the head is not at the specific position. Then each configuration of $M$ can be represented with a binary string $1s_0v_1s_1 ... v_ns_n$ such that $v_1,...,v_n\in\{0,1\}$ represent the cell content and $s_0,...,s_n\in L_{\mathtt{STATES}}$ represent the head and state information such that for one $i\in[n]_0$ we have $s_i = \enc(q_i)$ for some $q_i\in Q$ and for all other $j\in[n]_0$ with $j\neq i$ we have $s_j = 0^c$.

    Now, we will construct a 1-counter language $L$ which is a superset of all computations of $M$ in a specific encoded format. We will then construct a context-set $C$ based on $M$ which by right-greedy iterated bracketed contextual deletion results in a language $[L]_{\to_{[C]]}^*}$ whose intersection with a specific regular language is empty if and only if $L(M)$ is empty. As a consequence, if that emptiness was decidable, then we could decide the emptiness problem for $M$ and in particular for arbitrary Turing machines, which is a contradiction.
 
    We split up the construction of the 1-counter language $L$ into several parts and construct it from there on. In a first step, we want to obtain a language that contains words where the first/left part contains encoded information about a configuration of a Turing machine and the second/right part contains the reversed version of a configuration of a Turing machine. Using right-greedy bracketed contextual deletion, we can then determine words where the encoded bitstring of the left side matches the encoded bitstring of the right side, i.e., where the encoded bitstring of the right side is the reversal of the one from the left.
        
    First, we define two languages containing information about bit cells for encoded configurations in the left and right parts. From now on assume that any newly introduced letter is part of the alphabet $\Sigma'$ of the constructed 1-counter language. Also, we define a homomorphism $h : \{0,1\}^* \rightarrow \{0,1,@\}$ with $h(0) = @0$ and $h(1) = @1$ to enrich certain characters with an extra symbol used in the bracketed contextual deletion process later. Let 
    $$ L_{\mathtt{L\text{-}BC}} := \{\ (\$\tb)^n\ h(v)\ \ta^{2n+v}\ |\ n\in\N, v\in\{0,1\}\ \} $$
    be the language of left bit cells and let 
    $$ L_{\mathtt{R\text{-}BC}} := \{\ (\$\tb)^{2n+v}\ v\ \ta^{n}\ |\ n\in\N, v\in\{0,1\}\ \} $$
    be the language of right bit cells. Clearly both can be accepted by a deterministic 1-counter automaton. Using these, we define the language containing the head information for each cell, i.e., the blocks containing the encodings of either a state or $0^c$. Let
    $$ L_{\mathtt{L\text{-}HS}} := \{\ w\in (L_{\mathtt{L\text{-}BC}})^c\ |\ \pi_{\{0,1\}}=\enc(q_i) \text{ for some } q_i\in Q\ \} $$
    be the language of bit cells representing a head position with state information and let
    $$ L_{\mathtt{L\text{-}HE}} := \{\ w\in (L_{\mathtt{L\text{-}BC}})^c\ |\ \pi_{\{0,1\}}=0^c\ \} $$
    be the language of bit cells representing position where the head information is empty, i.e., the head is not there. Analogously, we define the same for the right versions of the languages and set
    $$ L_{\mathtt{R\text{-}HS}} := \{\ w\in (L_{\mathtt{R\text{-}BC}})^c\ |\ \pi_{\{0,1\}}=\enc(q_i)^R \text{ for some } q_i\in Q\ \} $$
    as well as
    $$ L_{\mathtt{R\text{-}HE}} := \{\ w\in (L_{\mathtt{R\text{-}BC}})^c\ |\ \pi_{\{0,1\}}=0^c \text{ for some } q_i\in Q\ \}. $$
    Additionally, we need some altered version of these state block representations which are used for the end of the tape of a configuration and its reversal. Here, the last bit cell (or first bit cell in the reversal/right side) is not of the same form as the others. We begin with the left side and define the empty and non empty head languages by
    \begin{align*} 
        L_{\mathtt{L\text{-}HS\text{-}TE}} := \{\ w = w'(\$\tb)^nh(v)\#\tc^{2n+v}\ |\ & w'\in(L_{\mathtt{L\text{-}BC}})^{c-1}, n\in\N, v\in\{0,1\},\\ 
                                                                            & \pi_{0,1}(w) = \enc(q_i) \text{ for some } q_i\in Q\ \} 
        \end{align*}
        as well as
        $$ L_{\mathtt{L\text{-}HE\text{-}TE}} := \{\ w = w'(\$\tb)^nh(0)\#\tc^{2n}\ |\ w'\in(L_{\mathtt{L\text{-}BC}})^{c-1}, n\in\N, \pi_{0,1}(w) = 0^c\ \}$$
        using the $\mathtt{TE}$ is the index to indicate the tape end. In a similar manner, we do the same for the right/reversed versions and define 
        \begin{align*} 
        L_{\mathtt{R\text{-}HS\text{-}TE}} := \{\ w = (\$\td)^{2n+v}\#v\ta^nw'\ |\ & w'\in(L_{\mathtt{R\text{-}BC}})^{c-1}, n\in\N, v\in\{0,1\},\\ 
                                                                            & \pi_{0,1}(w) = \enc(q_i)^R \text{ for some } q_i\in Q\ \} 
    \end{align*}
    as well as 
    $$ L_{\mathtt{R\text{-}HE\text{-}TE}} := \{\ w = (\$\td)^2n\#0\ta^nw'\ |\ w'\in(L_{\mathtt{R\text{-}BC}})^{c-1}, n\in\N, \pi_{0,1}(w) = 0^c\ \}. $$
    As $L_{\mathtt{L\text{-}BC}}$ and $L_{\mathtt{R\text{-}BC}}$ are recognized by 1-counter automata and as all other language are essentially extended forms of repeating $L_{\mathtt{L\text{-}BC}}$ or $L_{\mathtt{R\text{-}BC}}$ independently, we get that all the above languages must be 1-counter languages. Now, we can define the languages containing all left (and respectively right) bitstrings. An example afterwards clarifies the construction. Let
    \begin{align*}
            L_{\mathtt{LEFT}} := \{\ &h(1)\ta\ s_0\ \prod_{i = 1}^{n}(u_is_i)\ (\$\tb)^{k_r}h(v_r)\ta^{2k_r+v_r}\ s_{n+1}\ |\ v_r\in\{0,1\}, \\ 
            & k_r,n\in\N, u_1,...,u_n\in L_{\mathtt{L\text{-}BC}}, s_0,s_1,...,s_n\in(L_{\mathtt{L\text{-}HS}}\cup L_{\mathtt{L\text{-}HE}}) \\
            & s_{n+1}\in (L_{\mathtt{L\text{-}HS\text{-}TE}}\cup L_{\mathtt{L\text{-}HE\text{-}TE}}) \text{ and there exists one } j\in[n+1]_0 \\
            & \text{ such that } s_j\in (L_{\mathtt{L\text{-}HS}}\cup L_{\mathtt{L\text{-}HS\text{-}TE}}) \text{ and for all other } j'\in[n+1]_0 \\
            & \text{ with } j'\neq j \text{ we have } s_{j'} \in (L_{\mathtt{L\text{-}HE}}\cup L_{\mathtt{L\text{-}HE\text{-}TE}}) \}.
    \end{align*}
    To give an intuition on the words in $L_{\mathtt{LEFT}}$, consider the following example. Let $$w_b = (1\cdot10)(0\cdot00)(1\cdot00)$$ be some bitstring representing a configuration of some Turing machine. For each pair of brackets, the first number represents the symbol on the tape. The two numbers that follow represent the encoding of the state and head position. All but one pairs of brackets contain only $0^2$ to represent that the head is not at that position. In the remaining pair of brackets we see the number $10$ representing that the head is currently at that bit and that the current state is $q_2$. The encoding provided in $L_{\mathtt{LEFT}}$ would write $w_b$ in the following manner:
    $$ @1\ \ta(\$\tb)^{k_1}\ @1\ \ta^{2k_1+1}(\$\tb)^{k_2}\ @0\ \ta^{2k_2+0} \cdots (\$\tb)^{k_{n-1}}\ @0\ \ta^{2k_{n-1}+0}(\$\tb)^{k_n}\ @0\ \#\tc^{2k_n+0}  $$

    For the context-set $C$ which is later used to obtain the undecidable language, we add $(\ta\$,\tb)\in C$. Generalising the above structure we see that for each word $w\in L_{\mathtt{LEFT}}$ we have
    $$w = h(1)\ \ta(\$\tb)^{k_1}\ h(v_1)\ \ta^{2k_1+v_1}(\$\tb)^{k_2}\ ...\ (\$\tb)^k_n\ v_n\ \#\tc^{2k_n+v_n}.$$
    Notice that $w \to_{[C]]}^{*} w'$ for some word $w'\in (@0|@1)^*\#c^*$ if and only if $k_1 = 1$, $2k_i+v_i = k_{i+1}$ for all $i\in[n-1]$, and as a consequence $|w'|_c = \decimal(\pi_{\{0,1\}}(w))$. This fact is imporant to check whether the encoded right bitstring matches the reversal of the encoded left bitstring. Symmetrically to $L_{\mathtt{LEFT}}$, we now define $L_{\mathtt{RIGHT}}$.

    \begin{align*}
            L_{\mathtt{RIGHT}} := \{\ &s_0\ \ta^{2k_\ell+v_\ell}v_\ell\ta^{k_\ell}\ \prod_{i=1}^{n}(s_iu_i)\ s_{n+1}\ (\$\tb)1\ |\ v_\ell\in\{0,1\},  \\ 
            & k_\ell,n\in\N, u_1,...,u_n\in L_{\mathtt{R\text{-}BC}}, s_1,...,s_{n+1}\in(L_{\mathtt{R\text{-}HS}}\cup L_{\mathtt{R\text{-}HE}}) \\
            & s_0\in (L_{\mathtt{R\text{-}HS\text{-}TE}}\cup L_{\mathtt{R\text{-}HE\text{-}TE}}) \text{ and there exists one } j\in[n+1]_0 \\
            & \text{ such that } s_j\in (L_{\mathtt{R\text{-}HS}}\cup L_{\mathtt{R\text{-}HS\text{-}TE}}) \text{ and for all other } j'\in[n+1]_0 \\
            & \text{ with } j'\neq j \text{ we have } s_{j'} \in (L_{\mathtt{R\text{-}HE}}\cup L_{\mathtt{R\text{-}HE\text{-}TE}}) \}.
    \end{align*}
    Generally, we obtain that each word $w\in L_{\mathtt{RIGHT}}$ has the structure
    $$ w = (\$\td)^{k_n+v_n}\#\ v_n\ \ta^{k_n}\ \cdots\ (\$\tb)^{2k_1+v_1}\ v_1\ \ta^{k_1}(\$\tb)\ 1.$$
    Notice that the same as above holds: If we have $w \to_{[C]]}^{*} w'$ for some word $w'\in (\$\td)^*\#(0|1)^*$, then and only then we must have $k_1 = 1$, $2k_i+v_i = k_{i+1}$ for all $i\in[n-1]$, and as a consequence $|w'|_{\td} = \decimal(\pi_{\{0,1\}}(w)^R)$.

    Now, combining $L_{\mathtt{LEFT}}$ and $L_{\mathtt{RIGHT}}$, we notice the following. Let $w\in(L_{\mathtt{LEFT}}\cdot L_{\mathtt{RIGHT}})$. Then $w = w_\ell\#\tc^{k_\ell}(\$\td)^{k_r}\#w_r$ for some $k_\ell,k_r\in N$ and $w_\ell,w_r$ with $\#\notin(\alphabet(w_\ell)\cup\alphabet(w_r))$. As we have established before, if we can obtain a word $w \to_{[C]]}^{*} w'$ with $w' = w_\ell'\#\tc^{k_\ell}(\$\td)^{k_r}\#w_r'$ such that $\alphabet(w_\ell'w_r') = \{0,1,@\}$, then $|w'|_c = \decimal(\pi_{\{0,1\}}(w))$ and $|w'|_{\td} = \decimal(\pi_{\{0,1\}}(w)^R)$. For the construction of the context-set $C$ also add the context $(\tc\$,\td)\in C$. Then, we get that we can obtain $w \to_{[C]]}^{*} w_\ell'\#\#w_r'$ by splicing if and only if $k_\ell = k_r$ and in extension $\pi_{\{0,1\}}(w_\ell) = (\pi_{\{0,1\}}(w_r'))^R$.

    By that, it is ensured that we can only obtain a right side suffix that only contains the letters $0$ and $1$ and a left side prefix that only contains the letters $@0$ and $@1$ if and only if the binary word encoded on the right side matches the reversal of the encoded word on the left side. 
        
    As a final step, we construct the rest of $C$ to use it to check whether subsequent configurations are valid based on the definition of the transition function $\delta$ of $M$. For that, we consider the right side representation of one configuration and the left side representation of the next configuration. Let $w \in (L_{\mathtt{RIGHT}} \cdot L_{\mathtt{LEFT}})$ with $w = (\$\td)^{k_\ell}\#w_\ell w_r\#\tc^{k_r}$ and $\#\notin\alphabet(w_\ell w_r)$ such that $w \to_{[C]]}^{*} w'$ with $w' = (\$\td)^{k_\ell}\#w_\ell'w_r'\#\tc^{k_r}$
    with $\alphabet(w_\ell'w_r') = \{0,1,@\}$. We construct $C$ in a way such that $w_\ell' w_r' \to_{[C]]}^{*} \varepsilon$ if and only if $w_r$ represents a valid successor configuration to $w_\ell$. For that, we consider each tape blockwise from the middle on and splice the corrsponding blocks containing cell and head information if they correspond accordingly. 
        
    Notice that only in the middle between $w_\ell'$ and $w_r'$ we have that the symbol $@$ appears first and is preceded by multiple letters of $\{0,1\}$. After that, we always have that $@$ and symbols from $\{0,1\}$ alternate. This allows for each of the following constructions only to be used in that middle part where a suffix of $w_\ell$ and a prefix of $w_r$ have to be used for right-greedy bracketed contextual deletion, as in each suffix of $w_\ell$ at least two characters of $\{0,1\}$ appear and all following contexts use a left side with at least two symbols of $\{0,1\}$ before the first occurrence of $@$.
        
    (i) First, we construct $C_E = \{\ (\ 0^cv@\ ,\ vh(0^c)\ )\ |\ q_i\in Q, v\in\{0,1\}\}$. If on both sides the cell values equal and no state is encoded, i.e., the head is not there, then the contents of the same cell represented in $w_\ell'$ and $w_r$ have to match and can be spliced if that is the case.

    (ii) The second case considers transitions where the head is not moved. We construct $C_N := \{\ (\ \binary_c(q_i)^Rv_1@\ ,\ v_2h(\binary_c(q_j)\ ))\ |\ v_1,v_2\in\{0,1\}, q_i,q_j\in Q, (q_j,v_2,N)\in\delta(q_i,v_1) \}$ which allows for bracketed contextual deletion if the two blocks are valid successors regarding to $\delta$.

    (iii) In the third case, we have that the head moves in one direction. We construct $C_R$ and $C_L$ in the following manner, depending on the direction of the head of the tape. Remember that $w_\ell'$ contains an encoding of the reversal of one configuration and that $w_r$ contains an encoding of the potentially subsequent configuration which is not reversed. Hence, if we want to compare two different cells at once, we have to consider the reversed version on the other side. First, we set $C_R := \{\ (\ 0^cv_1\binary_c(q_i)^Rv_2@\ ,\ v_3h(0^c)@v_1h(\binary_c(q_j))\ )\ |\ v_1,v_2,v_3\in\{0,1\}, q_i,q_j\in Q, (q_j,v_3,R)\in\delta(q_i,v_2)\}$. As the head moves to the right in the non-reversed representation, the positions left to the head in the reversed representation have to be considered. Symmetrically, we can set $C_L := \{\ (\ \binary_c(q_i)^Rv_10^cv_2@\ ,\ v_2h(\binary_c(q_j))@v_3h(0^c)\ )$\\ $|\ v_1,v_2,v_3\in\{0,1\}, q_i,q_j\in Q, (q_j,v_3,L)\in\delta(q_i,v_1)\}$ to consider the transitions that move the head to the left.

    Now, set $C = \{(\ta\$,\tb),(\tc\$,\td)\} \cup C_E \cup C_N \cup C_R \cup C_L$. Using $C$, we can see that we can only obtain $w_\ell' w_r' \to_{[C]]}^{*} \varepsilon$ if and only if $w_r'$ encodes a valid successor of $(w_\ell')^R$ regarding the transition function $\delta$ of $M$ as we can only apply iterated right-greedy bracketed contextual deletion if and only if we can inductively confirm that the two encodings fit to each other, starting in the middle between $w_\ell$ and $w_r$, i.e., at the left end of the configuration encoding.
        
    We now define the 1-counter language $L$ from which we obtain that $[L]_{\to_{[C]]}^{*}}$ is undecidable. First, by the same argument as to why the languages at the beginning of the proof are 1-counter languages, we immediately obtain that $L_{\mathtt{LEFT}}$ and $L_{\mathtt{RIGHT}}$ must be 1-counter languages as well. The additional requirement here is to ensure that only one block may encode a head position. This can be done deterministically as we always know at which positions the head encoding happens and that all but one encodings of the head use only the word $0^c$.

    Before defining $L$, we define the $L_{\mathtt{LEFT\text{-}INIT}}$ and $L_{\mathtt{RIGHT\text{-}END}}$ to be adaptations of $L_{\mathtt{LEFT}}$ and $L_{\mathtt{RIGHT}}$ that can only use initial states (respectively final states) in their encodings. These adjustments can be directly handled via the statespace of some deterministic 1-counter automaton. Finally, set
    \begin{align*}
        L := \{\ u_1v_1u_2v_2 \cdots u_{n-1}v_{n-1}u_nv_n\ |\ &n\in\N, u_1\in L_{\mathtt{LEFT\text{-}INIT}}, v_n\in L_{\mathtt{RIGHT\text{-}END}} \\
                                                              &v_i\in L_{\mathtt{RIGHT}} \text{ for } i\in[n-1], \\ 
                                                              &u_j\in L_{\mathtt{LEFT}} \text{ for } j\in[n]\setminus\{1\}\ \}.
    \end{align*}
    As $L$ applies all the previously defined languages subsequently and independently, it directly follows that $L$ can be accepted using a 1-counter automaton. Whats left to show is that $[L]_{\to_{[C]]}^{*}}$ is undecidable. Define the regular language $L_r := (0|1|@)\#^*(0|1)$. Notice, that to obtain a word $w\in [L]_{\to_{[C]]}^{*}} \cap L_r$, all parts between $\#$ must be reduced to the empty word. 
        
    For that to happen between some words $u_iv_i$ for some $i\in[n]$, we know from before, first, that the encoded configuration in $u_i$ must be equal to the reverse of the encoded configuration in $v_i$ and that the number of $\tc$'s at the end of $u_i$ (respectively the number of $(\$\td)'s$ at the beginning of $v_i$) must correspond to the binary number contained in $u_i$ (respectively to the reversal of the binary number congainted in $v_i$). 
        
    Next, to reduce the words between $\#$ for some words $v_iu_i$ to the empty word, we have already shown before, that we must be able to apply the contexts in $C$ to compare and remove all blocks of cells encoded in $v_i$ and $u_i$ inductively. This can only be done, if we can apply some context until the part between the two $\#$'s is spliced to the empty word. By the definition of $C$, as given before, this only works if $u_i$ is an encoding of a valid succesor of $v_i$ and if the part between the $\#$'s only contains the letters $\{0,1,@\}$.

    Overall, using all previous arguments, we get that $[L]_{\to_{[C]]}^{*}} \cap L_r = \emptyset$ if and only if $L(M) = \emptyset$. Hence, if $[L]_{\to_{[C]]}^{*}} \cap L_r = \emptyset$ was decidable, then it could decide the emptiness problem for Turing machines as $M$ was assumed to be arbitrary. So, it must be undecidable. The intersection between some undecidable language and a regular languages is known to be undecidable. Hence, $[L]_{\to_{[C]]}^{*}}$ must also be undecidable.
\end{proof}

%% file: sec-5-conclusion.tex
This paper considered different approaches to formalize the process of co-transcriptional splicing in terms of formal language theory. For that, new context based deletion operations have been introduced. First, a version solely relying on the existence of contexts-called bracketed contextual deletion- and, second, a version relying on the formation of stable hairpins under different energy model considerations-called contextual lariat deletion- have been investigated. For bracketed contextual deletion, an iterated and a greedy version have been examined. For contextual lariat deletion, a linear energy model, a model assuming a constant bound for the loop-length, and a logarithmic energy model have been considered for languages obtained by applying single steps of contextual lariat deletion as well as languages obtained by applying contextual lariat deletion in a parallel manner.

First, the practically motivated Template Constructability Problem has been considered for bracketed contextual deletion as well as for contextual lariat deletion. Then, the computational power of both models have been thoroughly investigated. For that, language closure properties and connections to undecidable language classes have been drawn. 

In addition to these findings, we highlight several open questions that remain unresolved regarding both, the Template Constructibility Problem and the computational power of bracketed contextual deletion as well as contextual lariat deletion. These open questions pave the way for future investigations, aiming to further deepen our understanding of the complexities and potential applications of co-transcriptional splicing within formal models.

\subsection{Open Questions Regarding The Template Constructability Problem}

For the cases of greedy and non-greedy iterated bracketed contextual deletion, a comprehensive picture of the underlying complexities of all considered decision problems could be established. Verification of set obtainability from a given template as well as the rather restrictive Exact Template Constructability Problem have been shown to be solvable in polynomial time. The generalized Template Constructability Problem has been shown to be NP-complete in general. What remains to be proven, though, is whether Problem~\ref{problem:main-construction} is also NP-hard for binary and ternary alphabets.

\vspace{2mm}
\begin{question}
    Is the Template Constructability Problem $NP$-hard for alphabets $\Sigma$ with $|\Sigma| = 2$ or $|\Sigma| = 3$ for bracketed contextual deletion?
\end{question}
\vspace{2mm}

Very similar results could be shown for contextual lariat deletion under all three considered energy models. Here, the question regarding the NP-hardness of the Template Constructability Problem remains open for binary and ternary alphabets in the case of sets obtained under parallel deletion. For sets obtained only by the 1-step operation, however, it remains open whether the Template Constructability Problem might even be in $P$.

\vspace{2mm}
\begin{question}
    Is the Template Constructability Problem in $P$ under 1-step contextual lariat deletion? Is the Template Constructability Problem in $P$ under parallel contextual lariat deletion for alphabets $\Sigma$ with either $|\Sigma| = 2$ or $|\Sigma| = 3$? 
\end{question}

\subsection{Open Questions Regarding the Computational Power of Co-Transcriptional Splicing}

For iterated bracketed contextual deletion, a thorough picture of (non-)closure properties regarding the greedy variant could be obtained. For the non-greedy version that involves randomly skipping right contexts, however, many results could not be replicated, yet. Hence, for the classes of context-free languages $L_{cf}$, linear languages $L_{lin}$, and 1-counter languages $L_{1c}$, we propose the following open question.

\vspace{3mm}
\begin{question}
    Let $\mathcal{L}\in\{L_{cf}, L_{lin}, L_{1c}\}$ be one of these language classes. Is $\mathcal{L}$ closed under iterated bracketed contextual deletion?
\end{question}
\vspace{3mm}

For contextual lariat deletion, according to the literature on observed hairpins during co-transcriptional splicing, different energy models restricting the size of the loop of unbounded bases were considered. Those are either a linear bound of the loop with respect to the length of the stem, a constant size bound of the loop, or a logarithmic contribution of unbounded bases, resulting in an exponential bound of the length of the stem for the hairpin to be stable.

For the model that assumes a linear contribution of free energy regarding unbounded bases, i.e., the loop, similar to the results for iterated right-greedy bracketed contextual deletion, a thorough picture of (non-)closure properties was obtained and the undecidability of some decision problems has been established. But, even though we obtain some undecidability result, what is still left to be shown is whether we can obtain some undecidable language from, e.g., the sets obtained by 1-step or parallel contextual lariat deletion from a context-free or linear language in this setting.

\vspace{3mm}
\begin{question}
    Let $L$ be some context-free or linear language. Does there exist a configuration of parameters for $H_{lin}$ such that $[L]_{\hpdblino}$ or $[L]_{\hpdblinp}$ are undecidable languages?
\end{question}
\vspace{3mm}

For the model assuming a constant upper bound for the length of the loop of unbounded bases, similar (non-)closure results in comparison to the linear model have been established. What is still open, though, is whether linear languages are closed under the 1-step or parallel sets in this setting. Also, in contrast to the linear model, no undecidability results were established so far.

\vspace{3mm}
\begin{question}
    Let $\mathcal{L}_{lin}$ be the class of linear languages. Is $\mathcal{L}_{lin}$ closed under 1-step or parallel contextual lariat deletion assuming a constant bound on the loop length?
\end{question}
\vspace{3mm}

For the logarithmic energy model, similar non-closure properties to the ones of the linear energy model were obtained by analogous arguments. The existence of undecidability results remains an open question, though.

\vspace{3mm}
\begin{question}
    Does there exist some linear- or context-free language $L$ paired with parameters for $H_{log}$ such that $[L]_{\hpdblogo}$ or $[L]_{\hpdblogp}$ are undecidable?
\end{question}

Notice that any result regarding the closure of 1-counter languages under any energy model also remain an an open question for the moment.

\subsection{Final Remarks}

Both models, bracketed contextual deletion and contextual lariat deletion, can serve as a basis for future research, applying these models to simulate computations on RNA sequences. Depending on upcoming results, we hope to be able to simulate at least finite automata, if not even stronger models such as pushdown-automata or Turing machines with co-transcriptional splicing, using this formalism. All questions investigated in this paper can still be considered specifically for the case of the considered alphabet having size $4$, as this has most practical relevance due to $4$ being the size of the alphabets for DNA and RNA sequences. Finally, out of the perspective of formal language theory and combinatorics on words, finding answers for the open questions in the context of this paper might serve as an additional step towards finding practical solutions for the problems motivating this line of research in general.